\documentclass[12pt]{article}
\usepackage{footnote}
\usepackage{multirow}
\usepackage{url}
\usepackage[utf8]{inputenc}
\usepackage{mathtools}
\usepackage{algorithm}

\usepackage{amsbsy,amsmath,amsthm,amssymb,times}
\usepackage[longnamesfirst,sort]{natbib}
\usepackage{algpseudocode}
\usepackage{pslatex} 
\usepackage{commath}
\usepackage{bm}
\usepackage{caption}
\newtheorem{theorem}{Theorem}
\newtheorem{lemma}{Lemma}

\theoremstyle{definition}

\newtheorem{remark}{Remark}

\newcommand{\I}{{\mathrm{I}}}
\newcommand{\tr}{{\mathrm{tr}}}
\newcommand{\TV}{{\mathrm{TV}}}
\newcommand{\E}{{\mathrm{E}}}

\newcommand{\pcite}[1]{\citeauthor{#1}'s \citeyearpar{#1}}
\def\B{\mathcal{B}} 
\usepackage [english]{babel}
\usepackage [autostyle, english = american]{csquotes}
\MakeOuterQuote{"}

\allowdisplaybreaks

\begin{document}

\title{Block Gibbs samplers for logistic mixed models: convergence properties and a comparison with full Gibbs samplers
}

%\titlerunning{Block Gibbs samplers for logistic mixed models}        % if too long for running head

\author{Yalin Rao         and
  Vivekananda Roy\\
 % Department of Mathematics and Statistics,University of Massachusetts Amherst and
    Department of Statistics, Iowa State
  University, USA
}


% \institute{Yalin Rao \at
%               \email{yrao@iastate.edu}           %  \\
% %             \emph{Present address:} of F. Author  %  if needed
%            \and
%            Vivekananda Roy \at
%            \email{vroy@iastate.edu}\\
%               Department of Statistics, Iowa State University, USA
% }

 \date{}
%Received: date / Accepted: date}

\maketitle

\begin{abstract}
  The logistic linear mixed model (LLMM) is one of the most widely
  used statistical models. Generally, Markov chain Monte Carlo
  algorithms are used to explore the posterior densities associated
  with the Bayesian LLMMs. \pcite{polson2013bayesian} P\'{o}lya-Gamma
  data augmentation (DA) technique can be used to construct full Gibbs
  (FG) samplers for the LLMMs. Here, we develop efficient block Gibbs
  (BG) samplers for Bayesian LLMMs using the P\'{o}lya-Gamma DA
  method. We compare the FG and BG samplers in the context of a real
  data example, as the correlation between the fixed effects and the
  random effects changes as well as when the dimensions of the design
  matrices vary.  These numerical examples demonstrate superior
  performance of the BG samplers over the FG samplers. We also derive
  conditions guaranteeing geometric ergodicity of the BG Markov chain
  when the popular improper uniform prior is assigned on the
  regression coefficients, and proper or improper priors are placed on
  the variance parameters of the random effects. This theoretical
  result has important practical implications as it justifies the use
  of asymptotically valid Monte Carlo standard errors for Markov chain
  based estimates of the posterior quantities.
\end{abstract}
\section{Introduction}
\label{sec:int}
The logistic linear mixed model (LLMM) is an extensively used
generalized linear mixed model for binary response data. Let
$(Y_{1},Y_{2},...,Y_{n})$ denote the vector of Bernoulli
responses. Let $X$ and $Z$ be the $n \times p$ and $n \times q$ known
design matrices corresponding to the fixed effects and the random
effects, respectively. Suppose $x_{i}^\top$ and $z_{i}^\top$ indicate
the $i^{th}$ row of $X$ and $Z$, respectively, for $i = 1,...,n$. Let
$\beta \in \mathbb{R}^{p}$ be the regression coefficients vector and
$u \in \mathbb{R}^{q}$ be the random effects vector. In general, a
generalized linear mixed model (GLMM) can be built with a link
function that connects the probability that the response variable $Y$
equals 1 (that is, the expectation of $Y$) with $X$ and $Z$. For the
LLMM, $P(Y_{i} = 1) = F( x_{i}^\top \beta + z_{i}^\top u) $, where $F$
indicates the cumulative distribution function for the standard
logistic random variable, that is,
$F(t) = e^{t}/(1+e^{t}), \, t \in \mathbb{R}$. Also, assume that we
can divide $u$ into $r$ independent random effects. Let
$u = (u_{1}^\top,\dots,u_{r}^\top)^\top$, where $u_{j}$ is a
$q_{j} \times 1$ vector with $q_{j} > 0$, $j=1,\dots,r,$ and
$\sum_{j=1}^{r} q_{j} = q$. Assume,
$u_{j} \overset{ind}{\sim} N(0, (1/\tau_{j})\I_{q_{j}})$, where
$\tau_{j} > 0,$ for $j=1,\dots,r$. Let $\tau = (\tau_{1},...,\tau_{r})$. Thus the data
model for the LLMM is
\begin{align}
\label{eq:disu}
 Y_{i} \mid \beta,u,\tau &\overset{ind}{\sim} Ber(F(x_{i}^\top \beta + z_{i}^\top u))  \quad  \text{for} \quad i = 1,...,n, \nonumber \\ 
u_{j} \mid \tau_{j} &\overset{ind}{\sim} N(0, (1/\tau_{j})  \I_{q_{j}}), \quad j = 1,...,r.
\end{align} 

Let $y = (y_{1},y_{2},...,y_{n})$ denote the vector of observed Bernoulli responses. Then, the likelihood function for $(\beta,\tau)$ is 
\begin{align}
\label{eq:likef}
L(\beta,\tau \mid y)= \int_{\mathbb{R}^{q}}\prod_{i=1}^{n} \big{[}F(x_{i}^\top \beta + z_{i}^\top u)\big{]}^{y_{i}} \big[1- F(x_{i}^\top \beta + z_{i}^\top u) \big]^{1-y_{i}} \phi_{q}(u;0,D(\tau)^{-1}) du,
\end{align}
where $D(\tau)^{-1} = \oplus_{j=1}^{r} \frac{1}{\tau_{j}} \I_{q_{j}}$,
and $\oplus$ indicates the direct sum. Here,
$\phi_{q}(u;0,D(\tau)^{-1})$ denotes the probability density function
of the $q$-dimensional normal distribution with mean vector $0$,
covariance matrix $D(\tau)^{-1}$, and evaluated at $u$.

When working in a Bayesian framework, one must specify priors for $\beta$ and $\tau$. Here, we consider the prior for $\beta$ as given by
\begin{align}
\label{eq:betaprior}
\pi(\beta) \propto \exp \big{[}-\frac{1}{2}(\beta-\mu_{0})^\top Q(\beta-\mu_{0}) \big{]},
\end{align}
where $\mu_{0} \in \mathbb{R}^{p}$ and $Q$ is a $p \times p$ positive definite matrix (proper normal prior) or a zero matrix (improper uniform prior). Thus, if $Q = 0$, then $\pi(\beta) \propto 1 $. The prior for $\tau_{j}$ is 
\begin{align}
\label{eq:tauprior}
\pi(\tau_{j}) \propto \tau_{j}^{a_{j}-1}e^{-\tau_{j}b_{j}}, \;j = 1,...,r,
\end{align} which may be proper or improper depending on the values of $a_{j}$ and $b_{j}$. % (We will discuss the values of $a_{j}$ and $b_{j}$ in Section \ref{sec:ge}).
Finally, we assume that $\beta$ and $\tau$ are apriori independent, and all the $\tau_{j}$s are also apriori independent. Hence, the joint posterior density for $(\beta,\tau)$ is 
\begin{align}
\label{eq:jointp1}
\pi(\beta,\tau \mid y) = \frac{1}{c(y)}L(\beta,\tau \mid y)\pi(\beta)\pi(\tau),
\end{align}
where
$c(y) = \int_{\mathbb{R}_{+}^{r}} \int_{\mathbb{R}^{p}} L(\beta,\tau
\mid y)\pi(\beta)\pi(\tau) d\beta d\tau$ is the marginal pmf of
$y$. If $c(y)$ is finite, then the posterior density is proper. Since
we consider both proper and improper priors on $(\beta, \tau)$, if
improper priors are used, then $c(y)$ is not necessarily
finite. Conditions for the posterior propriety of nonlinear mixed models
with general link functions are given in \cite{chen2002necessary} and
\cite{wang2018convergence}. Theorem~\ref{theoremimproper}, given in
Section~\ref{sec:ge} of this paper, provides easily verifiable sufficient conditions
for propriety of \eqref{eq:jointp1} when $\pi(\beta) \propto 1 $.

Since the likelihood function $L(\beta,\tau \mid y)$ is not available
in a closed form, the posterior density for $(\beta,\tau)$ is not
tractable for any choice of priors on these parameters. Markov chain
Monte Carlo (MCMC) algorithms can be used to explore the posterior
density $\pi(\beta, \tau|y)$. Even in the absence of the random
effects, MCMC algorithms are generally used for exploring the
posterior densities corresponding to the basic logistic model or other
generalized linear models (GLMs). Using the data augmentation (DA)
technique \citep{vand:meng:2001}, in a highly cited paper,
\cite{albert1993bayesian} constructed a Gibbs sampler for GLMs with
the probit link. Since then there have been several attempts to
construct such a DA Gibbs sampler for the logistic model (see
e.g. \cite{holmes2006bayesian} and
\cite{fruhwirth2010data}). Recently, using the P\'{o}lya-Gamma (PG)
latent variables, \cite{polson2013bayesian} (denoted as PS\&W
hereafter) have proposed an efficient DA Gibbs sampler for the
Bayesian logistic models. A random variable $\omega$ follows a PG
distribution with parameters $a>0, \, b \ge 0$, if
$\omega \overset{d}{=} (1/(2\pi^2)) \sum_{i =1}^{\infty} g_{i} /
[(i-1/2)^2 +b^2/(4\pi^2)]$, where
$g_{i} \overset{iid}{\sim}$Gamma$(a,1)$. We denote this as
$\omega \sim $PG$(a,b)$. PS\&W's DA technique can be extended to
construct a Gibbs sampler for the LLMMs. Indeed, with PG latent variables
$\omega = (\omega_{1},\omega_{2},...,\omega_{n})$, one can construct a
joint posterior density $\pi(\beta,u,\omega,\tau \mid y)$ (details are
given in Section \ref{sec:BGFG}) such that \vspace*{-.06in}
\begin{align}
\label{eq:relationtwopos}
\int_{\mathbb{R}^{q}}\int_{\mathbb{R}_{+}^{n}} \pi(\beta,u,\omega,\tau \mid y) d \omega\, d u = \pi(\beta, \tau \mid y),
\end{align}
where $\mathbb{R}_{+} = (0,\infty)$, and $\pi(\beta, \tau \mid y)$ is
given in \eqref{eq:jointp1}. Using the conditional distributions of
the joint density $\pi(\beta,u,\omega,\tau \mid y)$, a full Gibbs
sampler can be formed (The details for this Gibbs sampler are given in
Section \ref{sec:FG}.). It is known that blocking
parameters can improve the performance of the Gibbs sampler in terms
of reducing its operator norm \citep{liu1994covariance}. In general,
when one or more variables are correlated, sampling them jointly can
improve the efficiency of MCMC algorithms
\citep{robe:sahu:1997, chib:rama:2010, ture:deva:paci:2017}. On the other hand, blocking may result in
complex conditional distributions that are not easy to sample
from. For the LLMMs, it turns out that an efficient two-block Gibbs
sampler can be constructed by using the two blocks, $\eta \equiv (\beta^\top,u^\top)^\top$ and $(\omega,\tau)$. We derive this block
Gibbs sampler in Section \ref{sec:blockgibbs}. Using numerical examples, we show that blocking can lead to great gains
in efficiency in Monte Carlo estimation for the LLMMs.

The block Gibbs Markov chain is Harris ergodic. Thus, the sample
(time) averages are consistent estimators of the means with respect to the
posterior density \eqref{eq:jointp1}. On the other hand, in practice,
it is important to ascertain the errors associated with the Monte Carlo
estimates. A valid standard error for the Monte Carlo estimate can be
formed if a central limit theorem (CLT) is available for the time
average estimator \citep{jones2001honest}. Establishing geometric
ergodicity (GE) of the underlying Markov chain is the most standard
method for guaranteeing CLTs hold for MCMC estimators, and is also
used for consistently estimating the asymptotic variance in the CLT
(\cite{vats2018strong}, \cite{vats2019multivariate}). GE of Gibbs
samplers for probit and logistic GLMs under different priors has been
established in the literature \citep{roy2007convergence,
  choi:hobe:2013, chak:khar:2017, wang2018geometric}. Also, GE of
Gibbs samplers for probit mixed models and normal linear mixed models
under improper priors on the regression coefficients and variance
components is considered in \cite{wang2018convergence} and
\cite{roman2012convergence}, respectively. From \cite{roy:2012a} it
follows that \pcite{wang2018convergence} GE result also holds for
parameter expansion for DA algorithms for the probit mixed
models. \cite{wang2018analysis} prove uniform ergodicity of a Gibbs
sampler for the LLMMs under a proper normal prior on $\beta$ and a
truncated proper prior on $\tau$. \cite{wang2018analysis} establish
a {\it minorization condition} to prove uniform ergodicity of the
Gibbs sampler \cite[see e.g.][Theorem
8]{roberts2004general}. \pcite{wang2018analysis} analysis of the Gibbs
Markov chain puts a constraint on the support of the posterior density
\eqref{eq:jointp1}, in that the parameter $\tau$ is bounded away from
zero. Here, using a {\it drift condition}, we establish geometric
convergence rates for the block Gibbs sampler, in the case when the
popular improper uniform prior is assigned on $\beta$ and the commonly used proper
gamma priors or the improper power priors are assigned on $\tau$. In
contrast to \cite{wang2018analysis}, our result does not put any
restriction on the support of the variance components, that is, the
common support of $\tau_j$'s is the entire positive real line.

 The rest of the article is organized as follows. In Section
\ref{sec:BGFG}, we provide details on PG data augmentation and
construct the full and the block Gibbs samplers. Section \ref{sec:numerical} contains
numerical examples. These examples are used to compare the performance
of the block and the full Gibbs samplers. In Section \ref{sec:ge}, we
consider geometric convergence of the block Gibbs sampler under improper priors. Some concluding remarks are
provided in Section \ref{sec:discussion}. Finally, several theoretical
results along with proofs of the results appear in the appendices.

\section{Gibbs samplers}
\label{sec:BGFG}
In this section, using the PG variables, we discuss DA for the LLMMs, and construct Gibbs samplers for \eqref{eq:jointp1}. Following \eqref{eq:likef} and \eqref{eq:jointp1}, the joint posterior density for $(\beta,\tau)$ is
\begin{align*}
\pi(\beta,\tau \mid y)= \frac{\pi(\beta)\pi(\tau)}{c(y)}\int_{\mathbb{R}^{q}}\prod_{i=1}^{n} \frac{\exp\{y_{i} (x_{i}^\top \beta + z_{i}^\top u)\}}{1+\exp(x_{i}^\top \beta + z_{i}^\top u)} \phi_{q}(u;0,D(\tau)^{-1})du. \end{align*} 
By Theorem 1 in \cite{polson2013bayesian}
\begin{align}
\label{eq:jointp2}
\pi(\beta,\tau \mid y) = \int_{\mathbb{R}^{q}}\int_{\mathbb{R}_{+}^{n}} \bigg[\prod_{i=1}^{n}& \frac{ \exp\{k_{i}(x_{i}^\top \beta + z_{i}^\top u)-\omega_{i}(x_{i}^\top\beta + z_{i}^\top u)^{2}/2\}}{2}p(\omega_{i})\bigg] d \omega \; \nonumber\\ & \times \phi_{q}(u;0,D(\tau)^{-1})du \times \frac{\pi(\beta)\pi(\tau)}{c(y)},
\end{align}
where $\omega=(\omega_{1},\omega_{2},...\omega_{n})$, $k_{i} = y_{i} -1/2, \, i = 1,...,n$ and $p(\omega_{i})$ is the pdf of $\omega_{i} \sim $PG$(1,0)$ given by, 
\begin{equation}
\label{eq:pg1}
p(\omega_{i})= \sum_{\ell = 0}^{\infty} (-1)^{\ell} \frac{(2\ell+1)}{\sqrt{2\pi \omega_{i}^{3}}} \exp\bigg[-\frac{(2\ell+1)^{2}}{8\omega_{i}}\bigg],\,\omega_{i} > 0.
\end{equation}
We now define the joint posterior density of $\beta,u,\omega,\tau$ given $y$ mentioned in \eqref{eq:relationtwopos} as 
\begin{align}
\label{eq:jointp3}
\pi(\beta,u,\omega,\tau \mid y) &\propto \bigg[\prod_{i=1}^{n}  \exp\{k_{i}(x_{i}^\top \beta + z_{i}^\top u) -\omega_{i}(x_{i}^\top \beta + z_{i}^\top u)^{2}/2\}p(\omega_{i}) \bigg] \pi(\beta)\pi(\tau) \nonumber\\ & \quad \times \phi_{q}(u;0,D(\tau)^{-1})  \nonumber \displaybreak\\
& = \bigg[\prod_{i=1}^{n}  \exp\{k_{i}(x_{i}^\top \beta + z_{i}^\top u)-\omega_{i}(x_{i}^\top \beta + z_{i}^\top u)^{2}/2\} p(\omega_{i})\bigg]
 \nonumber\\ & \quad \times \phi_{q}(u;0,D(\tau)^{-1}) \times \exp\Big[-\frac{1}{2}(\beta-\mu_{0})^\top Q(\beta-\mu_{0})\Big] \nonumber\\ & \quad\times \prod_{j=1}^{r} \tau_{j}^{a_{j}-1}\exp(-b_{j}\tau_{j})  ,
\end{align}
where \eqref{eq:jointp3} follows from the priors on $\beta$ and $\tau$ given in \eqref{eq:betaprior} and \eqref{eq:tauprior}, respectively.

\subsection{A full Gibbs sampler}
\label{sec:FG}
Let $\Omega(\omega)$ be the $n \times n$ diagonal matrix with the $i^{th}$ diagonal element $\omega_{i}$. Let $M=[X Z]$, and $m_{i}^\top$ indicates the $i^{th}$ row of $M$ for $i=1,\dots,n$. Let $\kappa = (k_{1},k_{2},...,k_{n})^\top$.  
We begin with deriving the conditional densities required for the full Gibbs (FG) sampler. Based on \eqref{eq:jointp3}, the conditional density of $\beta$ given $u, \omega, \tau, y$ is 
\begin{align*}
\pi(\beta \mid  u, \omega, \tau, y) &\propto \prod_{i=1}^{n} \exp\Big[k_{i} x_{i}^\top \beta-\omega_{i}(x_{i}^\top \beta)^{2}/2 - \omega_{i}x_{i}^\top \beta z_{i}^\top u \Big]\nonumber\\ & \quad \times \exp\Big[-\frac{1}{2}(\beta-\mu_{0})^\top Q(\beta-\mu_{0})\Big] \\
&\propto \exp \Big[ -\frac{1}{2}\beta^\top (X^\top \Omega(\omega) X + Q)\beta+ \beta^\top (X^\top \kappa + Q \mu_{0} - X^\top \Omega(\omega) Z u) \Big].
\end{align*}
Hence, 
\begin{align}
\label{eq:pbeta}
\beta \mid u,\omega, \tau, y \sim N((X^\top \Omega(\omega) X + Q)^{-1} (X^\top \kappa + Q \mu_{0} - X^\top \Omega(\omega) Z u),(X^\top \Omega(\omega) X+ Q)^{-1}).
\end{align}
Also from \eqref{eq:jointp3}, the conditional density of $u$ given $\beta, \omega, \tau, y$ is
\begin{align*}
\pi(u \mid \beta, \omega, \tau, y) &\propto \prod_{i=1}^{n} \exp\Big\{ k_{i} z_{i}^\top u - \frac{\omega_{i}}{2}\Big[ (z_{i}^\top u)^{2} + 2 z_{i}^\top u x_{i}^\top \beta\Big] \Big\} \exp\Big[ -\frac{1}{2}u^\top D(\tau)u \Big]\\
&= \exp \Big[  - \frac{1}{2} u^\top (Z^\top \Omega(\omega) Z + D(\tau)) u + u^\top (Z^\top \kappa - Z^\top \Omega(\omega) X \beta)\Big].
\end{align*}
Thus, it follows that
\begin{align}
\label{eq:pu}
u \mid \beta, \omega, \tau, y \sim  N((Z^\top \Omega(\omega) Z + D(\tau))^{-1}(Z^\top \kappa - Z^\top \Omega(\omega) X \beta), (Z^\top \Omega(\omega) Z + D(\tau))^{-1}).
\end{align}
Also from \eqref{eq:jointp3}, the conditional density of $\omega$ and $\tau$ given $\eta$ and $y$ is as follows
\begin{align}
\label{eq:pomegatau}
\pi(\omega,\tau \mid \eta,y) & \propto \prod_{i=1}^{n} \exp(-\omega_{i}(m_{i}^\top\eta)^{2}/2) p(\omega_{i}) \left|D(\tau)\right|^{\frac{1}{2}} \exp(-u^\top D(\tau)u/2) \nonumber\\ & \quad \times \prod_{j=1}^{r} \tau_{j}^{a_{j}-1}\exp(-b_{j}\tau_{j}),
\end{align}
where $\left|D(\tau)\right|$ is the determinant of $D(\tau)$. From the above, we see that $\omega_{i}$'s, $i=1,...,n$ are conditionally independent given $(\eta,y)$. Also, $(\omega, \tau)$ are conditionally mutually independent given $(\eta, y)$. The conditional density for $\omega_{i}$ is
\begin{align}
\label{eq:pomegai}
\pi(\omega_{i} \mid \eta,y)
& \propto \exp(-\omega_{i}(m_{i}^\top \eta)^{2}/2) p(\omega_{i}) \nonumber\\ &  
=\sum_{\ell=0}^{\infty}(-1)^{\ell} \frac{(2\ell+1)}{\sqrt{2\pi \omega_{i}^{3}}} \exp\Big(-\frac{(2\ell+1)^{2}}{8\omega_{i}} - \frac{\omega_{i}(m_{i}^\top\eta)^{2}}{2}\Big), 
\end{align}
where the equality follows from \eqref{eq:pg1}.
From \cite{wang2018geometric}, the pdf for PG$(a,b), a>0, \,b \in \mathbb{R}$ is 
\begin{align*}
p(x \mid a,b) = \Big[\cosh \big(\frac{b}{2}\big)\Big]^a\frac{2^{a - 1}}{\Gamma{(a)}} \sum_{\ell=0}^{\infty}(-1)^\ell \frac{\Gamma{(\ell+a)}}{\Gamma{(\ell+1)}}\frac{(2\ell+a)}{\sqrt{2\pi x^{3}}} \exp\Big(-\frac{(2\ell+a)^{2}}{8x} - \frac{xb^{2}}{2}\Big), %\quad 
\end{align*}
for $x>0,$ where the hyperbolic cosine function $\cosh(t) = (e^{t} + e^{-t})/2$. Hence, from \eqref{eq:pomegai} we have
\begin{align}
\label{eq:distributionomega}
\omega_{i} \mid \eta,y \overset{ind}\sim PG \big(1,\abs{m_{i}^\top \eta} \big), \,i = 1,....n.
\end{align}
From \eqref{eq:pomegatau}, the conditional density for $\tau_{j}$ is given by
\begin{align}
\label{eq:pdensitytau}
&\pi(\tau_{j} \mid \eta,y) \propto \tau_{j}^{q_{j}/2+a_{j}-1}\exp\big{[}-\tau_{j}(b_{j} + u_{j}^\top u_{j}/2)\big{]},\,j = 1,...,r.
\end{align}
Thus, we have
\begin{align*}
 \tau_{j} \mid \eta,y \overset{ind}\sim  \text{Gamma}(a_{j} + q_{j}/2,b_{j} + u_{j}^\top u_{j}/2),\,j = 1,...,r. 
 \end{align*}
We will allow $a_{j} > -q_{j} /2$, and $b_{j} \ge 0$. The case where $b_{j} = 0$ is more complicated and will be discussed at length in Section \ref{sec:ge}.

Let $(\beta^{(m)}, u^{m}, \omega^{(m)}, \tau^{(m)})$ denote the $m^{th}$ element for $(\beta, u, \omega, \tau)$ in the FG chain. Thus, a single iteration of the full Gibbs sampler $\{\beta^{(m)}, u^{m}, \omega^{(m)},\tau^{(m)}\}_{m=0}^{\infty}$ consists of the following four steps:
\begin{center}
\begin{table*}[h]
\begin{tabular}{l}
\hline
$\boldsymbol{Algorithm \, 1}$ \; The (m+1)st iteration of the full Gibbs sampler  \\  \hline
1. Draw $\tau_{j}^{(m+1)} \overset{ind}\sim$ Gamma$(a_{j} + q_{j}/2, b_{j} + u_{j}^\top u_{j}/2), \,j = 1,...,r$ with $u = u^{(m)}$. \\
2. Draw $\omega_{i}^{(m+1)} \overset{ind}{\sim} PG \big(1,\abs{m_{i}^\top\eta^{(m)}} \big),\, i = 1,....n$.\\
3. Draw $u^{(m+1)} \sim \eqref{eq:pu} % N((Z^\top \Omega(\omega)^{(m+1)} Z + D(\tau^{(m+1)}))^{-1} (Z^\top \kappa - Z^\top\Omega(\omega)^{(m+1)} X \beta^{(m)}), (Z^\top \Omega(\omega)^{(m+1)}Z + D(\tau^{(m+1)}))^{-1})
  $ with $\tau = \tau^{(m+1)},\,\omega = \omega^{(m+1)}$, and $\beta = \beta^{(m)}$. \\
4. Draw $\beta^{(m+1)} \sim \eqref{eq:pbeta} % N((X^\top \Omega(\omega)^{(m+1)} X + Q)^{-1} (X^\top \kappa+ Q \mu_{0} - X^\top\Omega(\omega)^{(m+1)} Z u^{(m+1)}), (X^\top \Omega(\omega)^{(m+1)}X + Q)^{-1})
  $ with $\omega = \omega^{(m+1)}$.\\
\hline
\end{tabular}
\end{table*}
\end{center}

\subsection{A two-block Gibbs sampler}
\label{sec:blockgibbs}
In this section, we construct a block Gibbs (BG) sampler for \eqref{eq:jointp1}. As mentioned before, $M = [X Z]$ with the $i^{th}$ row $m_{i}^\top$ for $i=1,\dots,n$. Note that $x_{i}^\top\beta + z_{i}^\top u = m_{i}^\top \eta$. From \eqref{eq:jointp3}, the conditional density of $\eta$ given $\omega,\tau,y$ is given by
\begin{align}
\label{eq:peta}
\pi(\eta \mid \omega,\tau,y) &\propto \prod_{i=1}^{n} \exp\big{[}k_{i}m_{i}^\top \eta-\omega_{i}(m_{i}^\top \eta)^{2}/2\big{]}
\exp\Big[-\frac{1}{2}u^\top D(\tau) u\Big]\nonumber\\ & \quad \times\exp\Big[-\frac{1}{2}(\beta-\mu_{0})^\top Q( \beta-\mu_{0})\Big] \nonumber\\
&\propto \exp\Big[-\frac{1}{2}(\eta - \Sigma (M^\top \kappa+ b))^\top \Sigma^{-1}(\eta - \Sigma (M^\top \kappa+b))\Big], 
\end{align}
where $\Sigma^{-1} = M^\top \Omega(\omega) M + A(\tau)$,
\[
  b_{(p+q)\times1} = \begin{pmatrix} Q\mu_{0} \\ 0_{q \times 1} \end{pmatrix}\; \mbox{and} \; A(\tau)_{(p+q)(p+q)} = \begin{pmatrix}   
                        Q  & 0 \\
                         0  & D(\tau)
                       \end{pmatrix}.
                     \]
                     Hence,
\begin{align}
\label{eq:distrieta}
\eta \mid \omega,\tau,y \sim N((M^\top \Omega(\omega) M + A(\tau))^{-1} (M^\top\kappa+b),(M^\top \Omega(\omega) M + A(\tau))^{-1}).
\end{align}
In the FG sampler in Section \ref{sec:FG}, $\tau, \, \omega, \, u$ and $\beta$ are drawn sequentially, whereas, in this section, we show that the conditional distribution of $\eta$ given $\omega,\,\tau,\,y$ is normal. From \eqref{eq:pomegatau}, we can see, conditional on $(\eta,y)$, $\omega$ and $\tau$ are independent. Thus, $\tau$ and $\omega$ can be drawn jointly as a block and we have a two-block Gibbs sampler.

Let $\eta^{(m)}, \omega^{(m)},$ and $\tau^{(m)}$ denote the values of $\eta, \omega,$ and $\tau$, respectively, in the $m^{th}$ iteration of the BG sampler. A single iteration of the block Gibbs sampler $\{\eta^{(m)}, \omega^{(m)},\tau^{(m)}\}_{m=0}^{\infty}$ consists of the following two steps:
\begin{center}
\begin{table*}[h]
\begin{tabular}{l}
\hline
$\boldsymbol{Algorithm \, 2}$ \; The (m+1)st iteration of the two-block Gibbs sampler  \\  \hline
1. Draw $\tau_{j}^{(m+1)} \overset{ind}\sim$ Gamma$(a_{j} + q_{j}/2, b_{j} + u_{j}^\top u_{j}/2), \,j = 1,...,r$ with $u = u^{(m)}$, \\
and independently draw $\omega_{i}^{(m+1)} \overset{ind}{\sim} PG \big(1,\abs{m_{i}^\top\eta^{(m)}}\big),\, i = 1,....n$.\\
2. Draw $\eta^{(m+1)} \sim \eqref{eq:distrieta} % N((M^\top \Omega^{(m+1)} M + A(\tau^{(m+1)}))^{-1} (M^\top\kappa + b), (M^\top \Omega^{(m+1)}M + A(\tau^{(m+1)}))^{-1})
  $ with $\tau = \tau^{(m+1)}$ and $\omega = \omega^{(m+1)}$.\\
\hline
\end{tabular}
\end{table*}
\end{center}
The conditional distribution of $\eta$ in the BG sampler and the
conditional distributions of $\beta$ and $u$ in the FG sampler are all
normal distributions of the form $N(S^{-1} t, S^{-1})$ for some
matrix $S$ and a vector $t$. Note that, for the conditional
distribution of $\eta$, $S$ is a $(p+q) \times (p+q)$ matrix, whereas
for $\beta$ and $u$ this is a $p \times p$ and $q \times q$ matrix,
respectively. Thus, a naive method of drawing from
$N(S^{-1} t, S^{-1})$ is inefficient especially if $p$ and/or $q$ is
large as it involves calculating inverse of the matrix $S$. Here, we
use a known method of drawing from $N(S^{-1} t, S^{-1})$ that does not
require computing $S^{-1}$. The method is as follows:
\begin{center}
\begin{table*}[h]
\begin{tabular}{l}
\hline
Algorithm for drawing from $N(S^{-1} t, S^{-1})$   \\  \hline
1. Let $S = LL^\top$ be the Cholesky decomposition of $S$.\\
2. Solve $L w = t$.\\
3. Draw  $z \sim N(0, \I_{k})$ where $k$ is the dimension of $S$.\\
4. Solve $L^\top x = w + z$. Then $x \sim N(S^{-1} t, S^{-1})$. \\
\hline
\end{tabular}
\end{table*}
\end{center}

\section{A real data example}
\label{sec:numerical}
We consider the student performance data set from
\cite{cort:silv:2008}. This data set includes $n=649$ observations and
$33$ variables including several categorical variables. As in
\cite{cort:silv:2008}, the binary response is defined as $1$ if the
final grade is greater than or equal to $10$, otherwise, it is defined
as $0$. Recall that $p$ denotes the number of columns for the design
matrix $X$. Also, note that, categorical variables are incorporated
into the LLMM as sets of dichotomous variables through what is known
as dummy coding. We consider different subsets of variables while
fitting the LLMM to compare the BG and FG samplers for different
dimensions. In particular, we consider $p=3, 7, 23$, including an
intercept term. We also keep one random effect `school' with $2$
levels in the LLMM. As mentioned in \cite{chib:rama:2010}, a key
principle to block sampling in MCMC is that `parameters in different
blocks are not strongly correlated whereas those within a block are'
\citep[see also][]{robe:sahu:1997, ture:deva:paci:2017}. As we observe
later in this Section, when $p$ varies, the average absolute (posterior)
correlations between $\beta$ and $u$ change greatly. The specific
values of $p$ we choose here are irrelevant for the general
conclusions of this section. Indeed, for some other $p$ values, close
to the ones we choose here, the general pattern of the different
empirical measures remains the same.

We analyze the data set by fitting the LLMM with a proper normal prior
\eqref{eq:betaprior} on $\beta$ with $\mu_{0} = 0 $ and
$Q =0.001 \I_{3} $. Also, we use a proper Gamma prior
\eqref{eq:tauprior} on $\tau_{1}$ where the (prior) mean and variance
of $\tau_1$ are $1.2$ and $100$, respectively ($a_{1} =0.0144 $ and
$b_{1} = 0.012$). We ran the BG sampler for $m=120,000$ iterations,
starting at an initial value $\eta^{(0)} = (\beta^{(0)}, u^{(0)})$
with a burn-in of $B=20,000$ iterations. Here, $\beta^{(0)}$ is the
estimate of $\beta$ obtained by fitting a logistic linear model
without any random effect. For $p=3, 7$, the initial value $u^{(0)}$
is a sample drawn from $N(0,(1/\tau_{1}^{(0)})\I_{2})$, where
$1/\tau_{1}^{(0)}$ is the estimate of the random effect variance
component obtained from the R package lme4. For $p=23$,
$1/\tau_{1}^{(0)}$ is the estimate of the random effect variance
component obtained from $p=7$ as lme4 did not run successfully in the
case of $p=23$. The FG sampler was also run for $m=120,000$ iterations
with a burn-in of $B=20,000$ iterations.

The BG and FG samplers are compared using the lag $k$ autocorrelation
function (ACF) values $k = 1,...,5$, the effective sample size (ESS) and
the multivariate ESS (mESS) (See \cite{roy:2020} for a simple introduction
to these convergence diagnostic measures.).  The ESS and mESS are
calculated using the R package mcmcse. We also compute the mean
squared jumps (MSJ) defined as
$\sum_{i=B+1}^{m} \norm{\beta^{(i+1)}-\beta^{(i)}}^2/(m-B)$ for the
$\beta$ variable, and similarly for the other variables. Here, $\norm{\cdot}$ denotes the Euclidean norm.
Tables~\ref{tab:realacf3}--\ref{tab:realacf23} provide the values of
the ACF for $\beta_0, \beta_1, \beta_2$ and $\tau_1$ for the BG and FG
samplers, as $p$ varies. Better performance of the BG sampler compared
to the FG sampler is observed from its mostly smaller ACF
values. Table~\ref{tab:realmess_3_7_23} provides the ESS values of
the intercept parameter, first two regression coefficients and
$\tau_1$. It also gives the mESS values for $u$ and $(\beta, \tau_1)$,
as $p$ varies. Again, better efficiency of the BG sampler compared to
the FG sampler is observed from its generally larger ESS and mESS values. We
calculate the average of the absolute (posterior) correlations between
the coordinates of the $\beta$ vector and those of the $u$ vector
computed based on the BG samples mentioned before. For $p=3, 7,$ and 23, these
values are 0.2601, 0.1143, and 0.0206, respectively. From the mESS
values of the parameters of interest, namely $(\beta, \tau_1)$, given
in Table~\ref{tab:realmess_3_7_23}, we see that the increase in
the efficiency of the BG sampler compared to the FG sampler is higher
for smaller $p$ values.  From Table~\ref{tab:realmsj}, it can be seen
that the BG sampler leads to higher MSJ values than the FG sampler
with the exception of $\tau_{1}$ in some cases. Thus,
Table~\ref{tab:realmsj} also corroborates better mixing of the BG
sampler than the FG sampler. So, in practice, the BG sampler can provide significant
gains compared to the FG sampler.

Finally, Table~\ref{tab:realmess_3_7_23} also provides the time
normalized efficiency (ESS and mESS values per second) for the two
samplers, as $p$ varies. From this Table, we see that the BG sampler
has always resulted in larger ESS and mESS values per second than the
FG sampler.  Indeed, the BG sampler results in higher time normalized
ESS values even in the cases when the FG sampler performs better in
terms of the ESS (See e.g. ESS ($\beta_1$) for $p=7$.). Recall that in
every iteration, the BG sampler makes a draw from a $(p+q)$
dimensional normal distribution, whereas the FG sampler draws from a
$q$ dimensional normal distribution and then a $p$ dimensional normal distribution. Other
draws are the same for both the BG and FG samplers. Using the Cholesky
decomposition method mentioned in Section \ref{sec:blockgibbs}, we
observe that for all values of $p$ considered here, the BG sampler
takes less time than the FG sampler to complete a certain number of
iterations. On the other hand, when $(p+q)$ takes much larger values,
the BG sampler takes more running time than the FG sampler.

\begin{center}
\begin{table}[h]
  \caption{ACF for the BG and FG samplers for the student performance data with $p=3$}
  \centering
\begin{tabular}{ccccccc}
\hline\hline Parameter&  Sampler  & lag 1& lag 2& lag 3& lag 4& lag 5\\
 \hline
$\beta_{0}$&BG&0.434&0.385&0.359&0.333&0.319\\
           &FG&0.985&0.974&0.964&0.956&0.948\\
\hline
$\beta_{1}$&BG&0.597&0.380&0.258&0.192&0.153\\
           &FG&0.613&0.402&0.282&0.212&0.173\\
           \hline
$\beta_{2}$&BG&0.836&0.734&0.667&0.621&0.586\\
           &FG&0.838&0.739&0.671&0.623&0.587\\
           \hline
$\tau_{1}$&BG&0.374&0.212&0.134&0.091&0.071\\
          &FG&0.405&0.285&0.224&0.186&0.155\\           
           \hline
\end{tabular}
\label{tab:realacf3}
\end{table}
\end{center}

\begin{center}
\begin{table}[h]
  \caption{ACF for the BG and FG samplers for the student performance data with $p=7$}
  \centering
\begin{tabular}{ccccccc}
\hline\hline Parameter&  Sampler  & lag 1& lag 2& lag 3& lag 4& lag 5\\
\hline
$\beta_{0}$&BG&0.463&0.409&0.365&0.340&0.320\\&FG&0.924&0.867&0.821&0.781&0.747\\
\hline
$\beta_{1}$&BG&0.436&0.191&0.087&0.036&0.018\\         &FG&0.437&0.193&0.088&0.038&0.017\\
           \hline
$\beta_{2}$&BG&0.419&0.185&0.085&0.047&0.031\\
           &FG&0.420&0.187&0.090&0.053&0.032\\
           \hline
$\tau_{1}$&BG&0.379&0.216&0.142&0.098&0.063\\
          &FG&0.372&0.244&0.190&0.150&0.122\\           
           \hline
\end{tabular}
\label{tab:realacf7}
\end{table}
\end{center}

\begin{center}
\begin{table}[h]
  \caption{ACF for the BG and FG samplers for the student performance data with $p=23$}
  \centering
\begin{tabular}{ccccccc}
\hline\hline Parameter&  Sampler  & lag 1& lag 2& lag 3& lag 4& lag 5\\
 \hline
$\beta_{0}$&BG&0.870 &0.811&0.761&0.714&0.669\\
           &FG&0.933&0.874&0.820&0.770&0.723\\
\hline
$\beta_{1}$&BG&0.637&0.437&0.323&0.256&0.213 \\      &FG&0.655&0.463&0.348&0.278&0.235\\
           \hline
$\beta_{2}$&BG&0.885&0.807&0.753&0.713&0.681\\
           &FG&0.880&0.801&0.744&0.702&0.669\\
           \hline
$\tau_{1}$&BG&0.384&0.228&0.147&0.098&0.070\\
          &FG&0.395&0.263&0.198&0.164&0.141\\           
           \hline
\end{tabular}
\label{tab:realacf23}
\end{table}
\end{center}

\begin{center}
\begin{table}[h]
  \caption{Multivariate ESS or ESS for the BG and FG samplers for the student performance data with $p=3, 7,$ and $23$. The numbers inside the parentheses are the corresponding values per second. The column name MC stands for Markov chain.}
 \centering
\begin{tabular}{ccccccccc}
\hline\hline $p$&  MC  & mESS ($\beta\; \tau_1$)   &mESS ($\beta$) & ESS ($\beta_{0}$) &ESS ($\beta_{1}$) &ESS ($\beta_{2}$)&mESS($u$) &ESS ($\tau_{1}$) \\
 \hline
\multirow{4}{*}{$3$}&\multirow{2}{*}{BG}&19,012 &15,979 &4,229 &9,268 &2,586 &34,623 &31,688 \\
&&(321) &(270)&(71)&(156)&(44)&(585)&(535)\\
           &\multirow{2}{*}{FG}&1,539 &978 &40 &8,089 &2,560 &76 &2,037 \\
 &&(22)&(14)&(1)&(115)&(36)&(1)&(29)\\
\hline
\multirow{4}{*}{$7$}&\multirow{2}{*}{BG}&27,474 &26,702 &4,494 &36,015 &31,294 &34,844 &32,334 \\
&&(426)&(414)&(70)&(559)& (486)&(541)& (502)\\ &\multirow{2}{*}{FG}&13,533 &12,793 &1,894 &38,572 &29,826 &1,931 &11,401 \\
&&(179)&(169)&(25)&(509)&(394)&(25)&(150)\\
           \hline
\multirow{4}{*}{$23$}&\multirow{2}{*}{BG}&23,068 &22,966 &3,031 &7,272 &2,130 &31,529 &28,561 \\
&&(295)&(294)&(39)&(93)&(27)&(404)&(366)\\
           &\multirow{2}{*}{FG}&18,016 &17,770 &3,040 &5,912 &1,554 &1,022 &7,911 \\
           &&(187)&(184)&(32)&(61)&(16)&(11)&
(82)\\           \hline
\end{tabular}
\label{tab:realmess_3_7_23}
\end{table}
\end{center}
\begin{center}
\begin{table}[h]
  \caption{Mean squared jumps for the BG and FG samplers for the student performance data with $p=3,7,$ and $23$}
  \centering
\begin{tabular}{cccc|ccc}
\hline\hline p& \multicolumn{3}{c}{BG}& \multicolumn{3}{c}{FG}\\
\hline
 &$\beta$&$u$&$\tau$& $\beta$&$u$&$\tau$\\
 \hline
3&12.94&24.41&1319.13&0.73&0.10&1071.40\\
\hline
7&15.52&24.97&1192.74&3.03&0.10&1309.20 \\
 \hline
 23&90.19&29.68 &1240.48&74.84&0.11&1249.00\\
 \hline
\end{tabular}
\label{tab:realmsj}
\end{table}
\end{center}

\section{Geometric ergodicity of the block Gibbs sampler}
\label{sec:ge} 
We begin this section with a discussion on the conditional density
$\pi(\tau \mid \eta,y)$. Since we allow the prior rate
parameter $b_{j}$ for $\tau_{j}$ to be zero, define
$A = \{j \in \{1,2,...,r\}: b_{j} = 0\}$. Recall from
\eqref{eq:pdensitytau} that
$\tau_{j} \mid \eta,y
\overset{ind}{\sim}$ Gamma$(a_{j} + q_{j}/2,b_{j} + u_{j}^\top
u_{j}/2),\,j = 1,...,r$ when $a_{j} + q_{j}/2 >0$ and
$b_{j} + u_{j}^\top u_{j}/2 >0$. The density
$\pi(\tau \mid \eta,y)=\prod_{j=1}^{r}\pi(\tau_{j} \mid
\eta,y)$ is not defined when $A$ is not empty and
$\| u_{j} \| = 0$ for $j \in A$. Let
$N = \{\eta \in \mathbb{R}^{p+q}, \prod_{j \in A} \| u_{j} \| =
0\}$. The fact that $\pi(\tau \mid \eta,y)$ is not
defined on $N$ is irrelevant for simulating the BG sampler as $N$ is a
measure zero set with respect to the Lebesgue measure on
$\mathbb{R}^{p+q}$. But, for a theoretical analysis of the BG chain,
$\pi(\tau \mid \eta,y)$ needs to be defined for all
$\eta \in \mathbb{R}^{p+q}$. %  Since the probability of $\eta$
% lying in $N$ is zero, the density $\pi(\tau \mid \eta,y)$ can be
% defined arbitrarily on $N$.
For all $\eta \in \mathbb{R}^{p+q}$,
we define
\begin{equation}
\label{eq:postau}
\pi(\tau \mid \eta, y) = \begin{cases} \prod_{j=1}^{r} f_{G} (\tau_{j}, a_{j}+\frac{q_{j}}{2}, b_{j} + \frac{1}{2}u_{j}^\top u_{j}) & \text{if} \quad \eta \not\in N\\ \prod_{j=1}^{r} f_{G} (\tau_{j}, 1, 1) & \text{if}\quad \eta \in N \end{cases} .
\end{equation} 
Here, $f_{G}(x, a, b)$ denotes the pdf of a gamma random variable with the shape parameter $a$, the rate parameter $b$, and evaluated at $x$. Thus, $f_{G}(x, a, b) = (b^a/\Gamma(a)) x^{a-1} \exp(-bx)$. The Markov transition density (Mtd) of the BG chain $\{\eta^{(m)}, \omega^{(m)},\tau^{(m)}\}_{m=0}^{\infty}$ is  
\begin{align}
\label{eq:mtd}
k(\eta,\omega,\tau \mid \eta',\omega',\tau') & =  \pi(\eta \mid \omega,\tau,y) \pi(\omega,\tau \mid \eta',y)\nonumber\\ & =  \pi(\eta \mid \omega,\tau,y) \bigg[\prod_{i=1}^n \pi(\omega_i \mid \eta',y)\bigg] \pi(\tau \mid \eta',y),
\end{align}
where the conditional densities $\pi(\eta \mid \omega,\tau,y), \pi(\omega_i \mid \eta',y),$ and $\pi(\tau \mid \eta',y)$ on the right side of
\eqref{eq:mtd} are given in \eqref{eq:peta}, \eqref{eq:pomegai}, and
\eqref{eq:postau}, respectively. It is easy to see that the joint density
\eqref{eq:jointp3} is the invariant density of $k$, and $k$ is
$\varphi$-irreducible. Thus, if \eqref{eq:jointp3} is a proper density,
that is, if $c(y) < \infty$ in \eqref{eq:jointp1}, then the BG
chain
$\{\eta^{(m)},
\omega^{(m)},\tau^{(m)}\}_{m=0}^{\infty}$ is Harris ergodic
\citep[][Chap 10]{meyn1993markov}, and hence, it
can be used to consistently estimate means with respect to \eqref{eq:jointp3}. Let
$S = \mathbb{R}^{p+q} \times \mathbb{R}_{+}^{n} \times
\mathbb{R}_{+}^{r}$. In fact, if $g:S \rightarrow \mathbb{R}$ is
integrable with respect to \eqref{eq:jointp3}, that is, if
\begin{align*}
\E_{\pi} |g(\eta,\omega,\tau)| := \int_{S}
\abs{g(\eta,\omega,\tau)}
\pi(\beta,u,\omega,\tau \mid y)d\eta \,d
\omega\, d\tau < \infty,
\end{align*}
then $\overline{g}_m := \sum_{i=0}^{m-1}
g(\eta^{(i)},\omega^{(i)},\tau^{(i)})/m \rightarrow
\E_{\pi}\,g$ almost surely as $m \rightarrow \infty$. On the other
hand, even when $\E_{\pi}\,g^2 < \infty$, Harris ergodicity of $k$ does
not guarantee a CLT holds for $\overline{g}_m $, which is used to obtain valid
standard errors of $\overline{g}_m $. We say a CLT for
$\overline{g}_m $ exists if
$\sqrt{m}( \overline{g}_m - \E_{\pi}\,g) \stackrel{d}{\rightarrow}
N(0,\sigma_{g}^2) \; as \;m \rightarrow \infty$ for some
$\sigma^2_{g} \in (0,\infty)$. Certain convergence rates of the BG
chain, as we explain next, ensure a CLT holds for $\overline{g}_m $.

Let $\B(S)$ denotes the Borel $\sigma$-algebra of $S$. Let $K^{(m)}: S \times \B(S) \rightarrow [0,1]$ denotes the $m$-step Markov transition function (Mtf) corresponding to the Mtd \eqref{eq:mtd}, that is, 
\begin{align*}
K^{(m)}((\eta',\omega',\tau'), B) = P ((\eta^{(m+j)},\omega^{(m+j)},\tau^{(m+j)}) \in B  \,|\, (\eta^{(j)},\omega^{(j)},\tau^{(j)})=(\eta',\omega',\tau')), 
\end{align*}
for any $j \in \{1, 2, \dots\}$ and for any measurable set $B \in \B(S)$. The BG chain is geometrically ergodic if there exist a function $H:S \rightarrow [0,\infty)$ and a constant $\rho \in (0,1)$ such that for all $m=0,1,2,...$,
\begin{align}
\label{eq:ge}
 \| K^m((\eta',\omega',\tau'),\cdot)-\Pi(\cdot)\|_{\TV}:=\sup_{B \in \B(S)} \| K^m((\eta',\omega',\tau'),B)-\Pi(B) \| \le H(\eta',\omega',\tau')\rho^m,
\end{align}
where $\Pi(\cdot)$ denotes the probability measure corresponding to the joint posterior density \eqref{eq:jointp3}, and $\| \cdot\|_{\TV}$ denotes the total variation norm. Harris ergodicity of $k$ implies the TV norm in \eqref{eq:ge} $\downarrow 0$ as $m \rightarrow \infty$, but does not ascertain any rate at which this convergence takes place. On the other hand, \eqref{eq:ge} guarantees a CLT for $\overline{g}_m $ if $\E_{\pi}\, |g|^{2+\delta} < \infty$ for some $\delta >0$ \citep{roberts2004general}. If \eqref{eq:ge} holds, it also implies that consistent batch means and a spectral variance estimator $\hat{\sigma}^2_{g}$ of $\sigma^2_{g}$ are available (\cite{vats2018strong}, \cite{vats2019multivariate}), and thus a valid standard error (SE) $\hat{\sigma}_{g}/\sqrt{m}$ for $\overline{g}_m $ can be calculated. An advantage of being able to calculate a valid SE is that it can be used to decide `when to stop' running the BG chain \citep{roy:2020}.

An important property that we are going to use in this article is that the marginal sequences  $\{\eta^{(m)}\}_{m=0}^{\infty}$, $\{(\omega^{(m)},\tau^{(m)})\}_{m=0}^{\infty}$ of the BG chain $\{\eta^{(m)},\omega^{(m)},\tau^{(m)}\}_{m=0}^{\infty}$ are themselves Markov chains, and either all three chains are geometrically ergodic or none of them \citep{liu1994covariance,robe:rose:2001}. Thus, we are free to analyze any of these chains to study their geometric convergence properties. Indeed, here we analyze the $\{\eta^{(m)}\}_{m=0}^{\infty}$ marginal chain.

We denote the Markov chain $\{\eta^{(m)}\}_{m=0}^{\infty}$ on $\mathbb{R}^{p+q}$ by $\Psi$ while the Markov chain $\{\eta^{(m)}\}_{m=0}^{\infty}$ on $\mathbb{R}^{p+q}\backslash N$ is denoted by $\tilde{\Psi}$. From \eqref{eq:mtd} it follows that the Mtd of the $\Psi$ chain is
\begin{align}
\label{eq:mtd_eta}
\tilde{k}(\eta \mid \eta') = \int_{\mathbb{R}_{+}^{r}} \int_{\mathbb{R}_{+}^{n}} \pi(\eta \mid \omega,\tau,y) \pi(\omega,\tau \mid \eta',y)d \omega \,d\, \tau,
\end{align}
\indent We can verify that $\tilde{k}(\eta \mid \eta')\pi(\eta' \mid y) = \tilde{k}(\eta' \mid \eta)\pi(\eta \mid y)$ for all $\eta, \eta' \in \mathbb{R}^{p+q}$ where $\pi(\eta \mid y) = \int_{\mathbb{R}_{+}^{r}} \int_{\mathbb{R}_{+}^{n}} \pi(\eta, \omega, \tau \mid y) d \omega\, d\, \tau$ is the $\eta$ marginal density of \eqref{eq:jointp3}. Hence, \eqref{eq:mtd_eta} is reversible with respect to  $\pi(\eta \mid y)$, and thus, $\pi(\eta \mid y)$ is the invariant density for the Markov chain $\{\eta^{(m)}\}_{m=0}^{\infty}$. Also, since $\{\eta^{(m)}\}_{m=0}^{\infty}$ is reversible, GE of the chain implies that CLTs hold for all square integrable functions with respect to $\pi(\eta \mid y)$ \citep{roberts1997geometric}. We first establish GE of the $\tilde{\Psi}$ chain. As explained in the proof of Theorem $1$, GE of $\tilde{\Psi}$ implies that of $\Psi$.  

\begin{theorem}
\label{theoremimproper}
If $\pi(\beta) \propto 1$, that is, if $Q = 0$ in \eqref{eq:betaprior}, the Markov chain underlying the block Gibbs sampler is geometrically ergodic if the following conditions hold:
\begin{enumerate}
\item $a_{j} < b_{j} = 0$ or $b_{j} > 0$ for $j = 1,...,r$;
\item    $a_{j} + q_{j}/2 > 0$ for $j = 1,...,r$;
\item    $M$ has full rank;
\item    There exists a positive vector $e>0$ such that $e^{\prime}M^{*} = 0$ where $M^{*}$ is an $n \times (p+q)$ matrix with $i$th row $c_{i}m_{i}^{\top}$, where $c_{i} =1-2 y_{i}$, $i=1,\dots,n$.
\end{enumerate}
    
\end{theorem}
\indent The proof of Theorem \ref{theoremimproper} is given in the Appendix C. The condition $4$ can be easily checked by an optimization method presented in \cite{roy2007convergence}.
\begin{remark}
The conditions in Theorem $1$ are the same as the conditions assumed in \pcite{wang2018convergence} Theorem $2$ that establishes GE of Gibbs samplers for the probit linear mixed model.
\end{remark}

\begin{remark}
  As mentioned before, \cite{wang2018analysis} analyzed the PG sampler
  for LLMMs with proper normal priors on $\beta$ and a truncated gamma
  prior on $\tau$. \pcite{wang2018analysis} proof established uniform
  ergodicity which is stronger than geometric ergodicity, but they
  assumed a stronger prior on $\tau$. Indeed, their proof involving a
  minorization condition requires that the support of $\tau$ is
  bounded away from zero. Our analysis of the BG Markov chain does not
  entail any minorization condition, and does not put any restriction, other than being positive,
  on the support of the variance components.
\end{remark}
\section{Conclusion}
\label{sec:discussion}
In this article, we consider an efficient block Gibbs sampler based on
the P\'{o}lya-Gamma DA \citep{polson2013bayesian} for one of the most
widely used statistical models, namely the LLMMs. Through numerical
examples, we observe that blocking can improve performance of the
P\'{o}lya-Gamma Gibbs samplers. We hope that the article will
encourage development and use of efficient blocking strategies for
Monte Carlo estimation of other GLMMs, including spatial GLMMs where
MCMC algorithms are known to suffer from slow mixing as noted in
\cite{evan:roy:2019}.

Undertaking
a Foster-Lyapunov drift analysis, we establish CLTs for the BG sampler
based MCMC estimators under the improper uniform prior on the regression
coefficients and improper or proper priors on the variance
components. These theoretical results are crucial for obtaining
standard errors for MCMC estimates of posterior means. In the process
of our proof for demonstrating CLTs for the BG sampler, we also
establish some general results on the P\'{o}lya-Gamma distribution. A
potential future problem is to construct and study block Gibbs
samplers for other GLMMs, including the mixed models with the robit
link \citep{roy2012convergence}.

\medskip
 \section*{Acknowledgment}
  The authors thank the editor and two anonymous reviewers for several helpful comments and suggestions that led to an improved revision of the paper.

\bigskip
\noindent {\Large \bf Appendices}

\begin{appendix}
\section{Some useful results}
Recall from Section \ref{sec:blockgibbs} that if $Q = 0$, that is, if $\pi(\beta)\propto 1$, then $b = 0$ and $A(\tau) =\begin{pmatrix}   
                        0  & 0 \\
                         0  & D(\tau)
                         \end{pmatrix}=  B(\tau),$ say.
In that case, \eqref{eq:distrieta} becomes
 \[\eta \mid \omega,\tau,y \sim N((M^\top \Omega(\omega) M + B(\tau))^{-1} M^\top\kappa,(M^\top \Omega(\omega) M + B(\tau))^{-1}).\]  By using the method of calculating the inverse of a partitioned matrix, the covariance matrix is
 \begin{align*}
 (M^\top \Omega(\omega) M + B(\tau))^{-1}&=\begin{pmatrix}
X^\top \Omega(\omega) X &X^\top \Omega(\omega) Z\\
Z^\top \Omega(\omega) X & Z^\top \Omega(\omega) Z + D(\tau)
\end{pmatrix}^{-1} \nonumber\\ & 
= \begin{pmatrix}
(\tilde{X}^\top\tilde{X})^{-1} + \tilde{R}\tilde{S}^{-1}\tilde{R}^\top &-\tilde{R}\tilde{S}^{-1}\\
-\tilde{S}^{-1}\tilde{R}^\top  &\tilde{S}^{-1}
\end{pmatrix},
 \end{align*}
 where $\tilde{X}=\Omega(\omega)^{\frac{1}{2}}X$, $\tilde{Z}=\Omega(\omega)^{\frac{1}{2}}Z$, $P_{\tilde{X}} = \tilde{X}(\tilde{X}^\top\tilde{X})^{-1}\tilde{X}^\top $, $\tilde{S} = \tilde{Z}^\top(I - P_{\tilde{X}})\tilde{Z} + D(\tau)$, and $\tilde{R}=(\tilde{X}^\top\tilde{X})^{-1}\tilde{X}^\top\tilde{Z}$. For the mean vector, it follows that
\begin{align}
\label{eq:meaneta}
(M^\top \Omega(\omega) M + B(\tau))^{-1} M^\top\kappa &=\begin{pmatrix}
(\tilde{X}^\top\tilde{X})^{-1} + \tilde{R}\tilde{S}^{-1}\tilde{R}^\top &-\tilde{R}\tilde{S}^{-1}\\
-\tilde{S}^{-1}\tilde{R}^\top  &\tilde{S}^{-1}
\end{pmatrix} \begin{pmatrix} X^\top\kappa \\ Z^\top \kappa  \end{pmatrix} \nonumber\\
&=\begin{pmatrix}
(\tilde{X}^\top\tilde{X})^{-1}X^\top\kappa + \tilde{R}\tilde{S}^{-1}\tilde{R}^\top X^\top\kappa-\tilde{R}\tilde{S}^{-1}Z^\top\kappa\\
-\tilde{S}^{-1}\tilde{R}^\top X^\top\kappa+\tilde{S}^{-1}Z^\top\kappa
 \end{pmatrix}.
\end{align} 
The first element in the right-hand side of \eqref{eq:meaneta} is the mean vector for $\beta$, while the second element in it is the mean vector for $u$. Thus,
\begin{align}
\label{eq:udistribution}
u \mid \omega,\tau,y \sim N(-\tilde{S}^{-1}\tilde{R}^\top X^\top\kappa+\tilde{S}^{-1}Z^\top\kappa,\tilde{S}^{-1}).
\end{align}
\begin{lemma}
\label{lemma:rineq}
Let $R_{j}$ be a $q_{j} \times q$ matrix consisting of $0$'s and $1$'s such that $R_{j}u = u_{j}$. Then 
\begin{align*}
(R_{j}\tilde{S}^{-1}{R_{j}}^\top)^{-1} \preceq \bigg(\sum_{i=1}^{n}\omega_{i} \tr(Z^\top Z) + \tau_{j} \bigg)\I_{q_{j}}.
\end{align*}
Here, for two matrices A and B, $A \preceq B$ means $B-A$ is a positive semidefinite matrix.  
\end{lemma}
\begin{proof} Let $\lambda_{max}$ denote the largest eigenvalue for $\tilde{Z}^\top(\I - P_{\tilde{X}})\tilde{Z}$, then
\begin{align*}
\tilde{S} = \tilde{Z}^\top(\I - P_{\tilde{X}})\tilde{Z} + D(\tau)
\preceq \lambda_{max} \I_{q} + D(\tau)
\preceq \tr(\tilde{Z}^\top(\I - P_{\tilde{X}})\tilde{Z})\I_{q} + D(\tau), 
\end{align*}
where $D(\tau) = \oplus_{j=1}^{r} \tau_{j} \I_{q_{j}}$ as defined before, and the second inequality follows from the fact that $\tilde{Z}^\top(\I - P_{\tilde{X}})\tilde{Z}$ is a positive semidefinite matrix. Now
\begin{align*}
  \tr(\tilde{Z}^\top (\I - P_{\tilde{X}})\tilde{Z}) \le \tr(\tilde{Z}^\top \tilde{Z})&=\tr(Z^\top \Omega(\omega) Z)\\
                                                                                     &=\tr\bigg(\sum_{i=1}^{n}\omega_{i}z_{i}z_{i}^\top\bigg)\\
&=\sum_{i=1}^{n} \tr(\omega_{i}z_{i}z_{i}^\top)=\sum_{i=1}^{n} \omega_{i} \tr(z_{i}z_{i}^\top)
\le \sum_{i=1}^{n} \omega_{i} \tr(Z^\top Z),  
\end{align*}
where $z_{i}^\top$ denotes the $i^{th}$ row of the $Z$ matrix, and the first inequality is due to the fact that $\tilde{Z}^\top P_{\tilde{X}}\tilde{Z}$ is a positive semidefinite matrix. Thus $\tilde{S} \preceq \sum_{i=1}^{n} \omega_{i} \tr(Z^\top Z)\I_{q} + D(\tau)$. Hence, $\tilde{S}^{-1} \succeq (\sum_{i=1}^{n}\omega_{i} \tr(Z^\top Z)\I_{q} + D(\tau))^{-1}$. Recall that $R_{j}u =u_{j}$. Extracting the result of the $j^{th}$ random effect, we obtain: \begin{align*}
R_{j}\tilde{S}^{-1}{R_{j}}^\top \succeq R_{j}\bigg(\sum_{i=1}^{n}\omega_{i} \tr(Z^\top Z)\I_{q} + D(\tau)\bigg)^{-1} {R_{j}}^\top = \bigg(\sum_{i=1}^{n}\omega_{i} \tr(Z^\top Z) + \tau_{j}\bigg)^{-1}\I_{q_{j}}. 
\end{align*}
Thus, we have 
%\begin{align*}
$(R_{j}\tilde{S}^{-1}{R_{j}}^\top)^{-1} \preceq (\sum_{i=1}^{n}\omega_{i} \tr(Z^\top Z) + \tau_{j})\I_{q_{j}}.$
%\end{align*}
\end{proof}

\section{Some properties of the P\'{o}lya-Gamma distributions}
\begin{lemma}
\label{le:eomega}
Suppose $\omega \sim PG(a,b)$.
\begin{enumerate}
\item If $a \ge 1, \, b \ge 0$, then for $0<s \le 1$,
$\E(\omega^{-s}) \le 2^{s}b^{s} + L(s)$, where $L(s)$ is a constant depending on $s$.
\item If $a < 1, \, b \ge 0$, then for $0<s < a$,
$\E(\omega^{-s}) \le 2^{-s} (\pi^2 + b^2)^s \frac{\Gamma(a-s)}{\Gamma(a)}$.
\end{enumerate}
\end{lemma}
\begin{proof} 
We first prove part $1$ for $a =1$. The probability density function of a $PG(1,b)$ random variable is 
\begin{align*}
f(x \mid 1,b) = \cosh(b/2) \sum_{\ell=0}^{\infty}(-1)^\ell \frac{(2\ell+1)}{\sqrt{2\pi x^{3}}}\exp\Big[-\frac{(2\ell+1)^{2}}{8x}-\frac{b^{2}}{2}x \Big], \, x > 0.
\end{align*}
We consider the two cases $b=0$ and $b > 0$ separately.\\
\textbf{Case 1}: $b=0$. Since $0<s \le 1$, for any $x>0$, we have $x^{-s} \le x^{-1} + 1$. Then,
\begin{align*}
\E(\omega^{-s}) \le \int_{0}^{\infty} (x^{-1} + 1) f(x \mid 1,0)dx = \int_{0}^{\infty} x^{-1} f(x \mid 1,0) dx + 1. 
\end{align*}
Now, 
\begin{align}
\label{eq:integx}
\int_{0}^{\infty} x^{-1} f(x \mid 1,0) dx &= \int_{0}^{\infty} x^{-1} \sum_{\ell=0}^{\infty}(-1)^\ell \frac{(2\ell+1)}{\sqrt{2\pi x^{3}}}\exp\Big[-\frac{(2\ell+1)^{2}}{8x} \Big]dx\nonumber\\
&=\int_{0}^{\infty} \sum_{\ell=0}^{\infty}(-1)^\ell x^{-\frac{5}{2}} \frac{(2\ell+1)}{\sqrt{2\pi}}\exp\Big[-\frac{(2\ell+1)^{2}}{8x} \Big] dx.
\end{align}
Let $h_{1}(x,\ell) = (-1)^\ell x^{-\frac{5}{2}} \frac{(2\ell+1)}{\sqrt{2\pi}}\exp\big{[}-\frac{(2\ell+1)^{2}}{8x} \big{]}$, then 
\begin{align*}
\sum_{\ell=0}^{\infty} \int_{0}^{\infty} \abs{h_{1}(x,\ell)} dx &= \sum_{\ell=0}^{\infty}  \frac{(2\ell+1)}{\sqrt{2\pi}} \int_{0}^{\infty} x^{-\frac{5}{2}} \exp\Big{[}-\frac{(2\ell+1)^{2}}{8x} \Big{]} dx\\& =8 \sum_{\ell=0}^{\infty} \frac{1}{(2\ell+1)^{2}} < \infty.
\end{align*}
Hence, $\abs{h_{1}}$ is integrable with respect to the product measure of the counting measure and the Lebesgue measure. By Fubini's Theorem, from \eqref{eq:integx} we have
\begin{align}
\label{eq:fubini}
\int_{0}^{\infty} x^{-1} f(x \mid 1,0) dx &=  \sum_{\ell=0}^{\infty}(-1)^\ell \frac{(2\ell+1)}{\sqrt{2\pi}} \int_{0}^{\infty} x^{-\frac{5}{2}} \exp\Big[-\frac{(2\ell+1)^{2}}{8x} \Big] dx \nonumber \\ &= 8 \sum_{\ell=0}^{\infty}(-1)^\ell (2\ell+1)^{-2} =8C,
\end{align} 
where C is Catalan's constant. Hence, $\E(\omega^{-s}) \le 8C + 1$.\\
\textbf{Case 2}: $b > 0$. Note that
\begin{align}
\label{eq:expectomega}
\E(\omega^{-s}) &= \int_{0}^{\infty} x^{-s} f(x \mid 1,b) dx \nonumber \\ &= \int_{0}^{\infty} x^{-s-\frac{3}{2}} \cosh(b/2) \sum_{\ell=0}^{\infty}(-1)^\ell \frac{(2\ell+1)}{\sqrt{2\pi }} \exp\Big{[}-\frac{(2\ell+1)^{2}}{8x}-\frac{b^{2}}{2}x \Big{]}dx.
\end{align}
According to $10.32.10$ in \cite{olver2010math}, we have
\begin{align}
\label{eq:intebessel}
\int_{0}^{\infty} x^{-s-\frac{3}{2}} \exp\Big[-\frac{(2\ell+1)^{2}}{8x} -\frac{b^{2}}{2}x     \Big] dx  = 2 K_{s+\frac{1}{2}}\Big{(}\frac{b(2\ell+1)}{2}\Big{)} \Big{(}\frac{2b}{2\ell+1}\Big{)}^{s+\frac{1}{2}},
\end{align}
where $K_{v}(\cdot)$ is the modified Bessel function of the second kind of order $v$. For $x > 0$, according to $10.32.8$ in \cite{olver2010math},
\begin{align}
K_{s+\frac{1}{2}}(x)  &= \frac{\sqrt{\pi}(\frac{1}{2}x)^{s+\frac{1}{2}}}{\Gamma(s+1)} \int_{1}^{\infty} e^{-xt}(t^{2}-1)^{s}dt  \nonumber\\
&= \frac{\sqrt{\pi}(\frac{1}{2}x)^{s+\frac{1}{2}}}{\Gamma(s+1)} e^{-x}\int_{0}^{\infty} e^{-xt}(t^{2}+2t)^{s}dt  \label{eq:besselupper1} \\
&\le \frac{\sqrt{\pi}(\frac{1}{2}x)^{s+\frac{1}{2}}}{\Gamma(s+1)} e^{-x}\int_{0}^{\infty} e^{-xt}(t^{2s}+2^{s}t^{s})dt  \nonumber\\
& =  \frac{\sqrt{\pi}(\frac{1}{2}x)^{s+\frac{1}{2}}}{\Gamma(s+1)} e^{-x}\Big{(} \frac{\Gamma(2s+1)}{x^{2s+1}} + 2^{s}\frac{\Gamma(s+1)}{x^{s+1}} \Big{)} \nonumber\\
&=\sqrt{\pi} e^{-x} \Big{[}  \frac{\Gamma(2s+1)}{\Gamma(s+1)}2^{-s-1/2}x^{-s-1/2} + 2^{-1/2}x^{-1/2} \Big{]}. \label{eq:besselupper}
\end{align}
Also, from \eqref{eq:besselupper1} we have
\begin{align}
\label{eq:bessellower}
K_{s+\frac{1}{2}}(x) \ge \frac{\sqrt{\pi}(\frac{1}{2}x)^{s+\frac{1}{2}}}{\Gamma(s+1)} e^{-x}\int_{0}^{\infty} e^{-xt} 2^{s}t^{s}dt =\sqrt{\pi} e^{-x} 2^{-1/2}x^{-1/2}.
\end{align}
Let $h_{2}(x,\ell) =  x^{-s-3/2} \cosh(b/2)(-1)^\ell \frac{(2\ell+1)}{\sqrt{2\pi }}\exp\big{[}-\frac{(2\ell+1)^{2}}{8x}-\frac{b^{2}}{2}x \big{]}$, then
%\newpage
\begin{align*}
\sum_{\ell=0}^{\infty} \int_{0}^{\infty} \abs{h_{2}(x,\ell)} dx &= \sum_{\ell=0}^{\infty} \cosh(b/2) \frac{(2\ell+1)}{\sqrt{2\pi}} \int_{0}^{\infty} x^{-s-\frac{3}{2}} \exp\Big[-\frac{(2\ell+1)^{2}}{8x} -\frac{b^{2}}{2}x     \Big] dx \\
&= \cosh(b/2) \sum_{\ell=0}^{\infty} \frac{(2\ell+1)}{\sqrt{2\pi}} 2 K_{s+\frac{1}{2}}\Big{(}\frac{b(2\ell+1)}{2}\Big{)} \big{(}\frac{2b}{2\ell+1}\big{)}^{s+\frac{1}{2}} \\
&\le 2 \cosh(b/2) \sum_{\ell=0}^{\infty} \frac{(2\ell+1)}{\sqrt{2\pi}} \sqrt{\pi} e^{-\frac{b(2\ell+1)}{2}}  \Big[  \frac{\Gamma(2s+1)}{\Gamma(s+1)}2^{-s-1/2} \\ & \hspace{.25in} \times \big(\frac{b(2\ell+1)}{2} \big)^{-s-1/2} + 2^{-1/2} \big(\frac{b(2\ell+1)}{2}\big{)}^{-1/2}\Big]\Big(\frac{2b}{2\ell+1}\Big)^{s+\frac{1}{2}} \\
&=2^{s} (1+e^{-b})\Bigg[ \sum_{\ell=0}^{\infty} \frac{e^{-b\ell}}{(2\ell+1)^{2s}} \frac{\Gamma(2s+1)}{\Gamma(s+1)} + \sum_{\ell=0}^{\infty}\frac{e^{-b\ell}}{(2\ell+1)^{s}}b^{s}      \Bigg] < \infty.
\end{align*}
The second equality follows by using \eqref{eq:intebessel}. The inequality is based on \eqref{eq:besselupper}.
The convergence of the two series in the last step can be obtained by utilizing the ratio test. Hence, $\abs{h_{2}}$ is integrable with respect to the product measure of the counting measure and the Lebesgue measure. By Fubini's Theorem and \eqref{eq:intebessel}, \eqref{eq:expectomega} becomes
\begin{align}
\label{eq:expectomega2}
\E(\omega^{-s}) =  \cosh(b/2) \sum_{\ell=0}^{\infty}(-1)^\ell \frac{(2\ell+1)}{\sqrt{2\pi }} 2 K_{s+\frac{1}{2}}\Big(\frac{b[2\ell+1]}{2}\Big) \Big(\frac{2b}{2\ell+1}\Big)^{s+\frac{1}{2}}.
\end{align}
When $\ell$ is even, applying \eqref{eq:besselupper} to \eqref{eq:expectomega2}, and when $\ell$ is odd, applying \eqref{eq:bessellower} to \eqref{eq:expectomega2}, we obtain
\begin{align}
\label{eq:expectomega3}
\E(\omega^{-s}) &\le 2 \cosh(b/2) \Bigg\{\sum_{even \; \ell} \frac{(2\ell+1)}{\sqrt{2\pi }} \sqrt{\pi} e^{-\frac{b(2\ell+1)}{2}} 
                  \Big[  \frac{\Gamma(2s+1)}{\Gamma(s+1)}2^{-s-1/2} \Big\{\frac{b(2\ell+1)}{2} \Big\}^{-s-1/2}  \nonumber\\ &   \quad + (b(2\ell+1))^{-1/2}  \Big] - \sum_{odd \; \ell} \frac{(2\ell+1)}{\sqrt{2\pi }} \sqrt{\pi} e^{-\frac{b(2\ell+1)}{2}} 2^{-1/2} \Big(\frac{b(2\ell+1)}{2}\Big)^{-1/2}  \Bigg\}\nonumber\\ &  \quad  \quad \times \bigg(\frac{2b}{2\ell+1}\bigg)^{s+\frac{1}{2}}\nonumber \\
  &=e^{b/2}(1+e^{-b})\bigg\{\sum_{even \; \ell} (2\ell+1)^{-2s} 2^s e^{-b(\ell+1/2)} 
\frac{\Gamma(2s+1)}{\Gamma(s+1)}  \nonumber\\ & \quad + \sum_{even \; \ell} (2\ell+1)^{-s} (2b)^s e^{-b(\ell+1/2)} - \sum_{odd \; \ell} (2\ell+1)^{-s} (2b)^s e^{-b(\ell+1/2)}  \bigg\}\nonumber \\
&=(1+e^{-b})  b^{s} \sum_{\ell=0}^{\infty} (-e^{-b})^{\ell}(\ell+1/2)^{-s} + (1+e^{-b})2^{-s} \frac{\Gamma(2s+1)}{\Gamma(s+1)} \nonumber\\ & \quad \quad \times \sum_{even \; \ell}e^{-b\ell} (\ell+1/2)^{-2s} \nonumber\\
&= (1+e^{-b})b^{s} \Phi(-e^{-b},s,1/2) + (1+e^{-b})2^{-s} \frac{\Gamma(2s+1)}{\Gamma(s+1)} \sum_{k=0}^{\infty} e^{-2bk} (2k+1/2)^{-2s}  \nonumber\\
&= (1+e^{-b})\frac{b^{s}}{\Gamma(s)} \int_{0}^{\infty} \frac{t^{s-1}e^{-t/2}}{1+e^{-b-t}} dt + (1+e^{-b})2^{-s} \frac{\Gamma(2s+1)}{\Gamma(s+1)}  \nonumber\\  & \quad \quad \times \sum_{k=0}^{\infty} e^{-2bk} (2k+1/2)^{-2s}, 
\end{align}
where $\Phi(\cdot)$ is the Lerch transcendent function. \\
\indent For fixed $s>0$, let 
\begin{align*}
f(b)&= (1+e^{-b})\frac{b^{s}}{\Gamma(s)} \int_{0}^{\infty} \frac{t^{s-1}e^{-t/2}}{1+e^{-b-t}} dt -2^s b^s \nonumber \\ &= \frac{b^{s}}{\Gamma(s)} \Big[ (1+e^{-b}) \int_{0}^{\infty}  \frac{t^{s-1} e^{-t/2}}{1+e^{-b-t}}dt -\int_{0}^{\infty}  t^{s-1} e^{-t/2} dt\Big] \nonumber \\ &=\frac{b^{s} e^{-b}}{\Gamma(s)} \int_{0}^{\infty}  \frac{(1-e^{-t})}{1+e^{-b-t}}  t^{s-1} e^{-t/2} dt.\nonumber
\end{align*}
Since $(1-e^{-t})  t^{s-1} e^{-t/2}/(1+e^{-b-t}) \le    t^{s-1} e^{-t/2}$, which is integrable, by the Dominated Convergence Theorem (DCT), it follows that $f(b)$ is a continuous function of $b$. Another application of DCT shows that
\begin{align*}
\lim_{b \to \infty} \int_{0}^{\infty}  \frac{1-e^{-t}}{1+e^{-b-t}}  t^{s-1} e^{-t/2} dt  =\int_{0}^{\infty} (1-e^{-t})  t^{s-1} e^{-t/2} dt \le \int_{0}^{\infty} t^{s-1} e^{-t/2} dt =2^s \Gamma(s). 
\end{align*}
Hence, $\lim_{b \to \infty} f(b) = 0$. Since $f(b)$ is a continuous function of $b$, $f(0)=0$ and $\lim_{b \to \infty} f(b) = 0$, we can conclude that $\abs{f(b)}$ can be bounded by a positive constant $f_{0}$, hence, 
\begin{align}
\label{eq:posupperbound}
(1+e^{-b})\frac{b^{s}}{\Gamma(s)} \int_{0}^{\infty} \frac{t^{s-1}e^{-t/2}}{1+e^{-b-t}} dt \le 2^{s} b^{s} + f_{0}.
\end{align}
As for the second term in \eqref{eq:expectomega3}, we have
\begin{align}
\label{eq:inequality}
& (1+e^{-b})2^{-s} \frac{\Gamma(2s+1)}{\Gamma(s+1)} \sum_{k=0}^{\infty} e^{-2bk} (2k+1/2)^{-2s} \nonumber\\ &\le (1+e^{-b})2^{-s} \frac{\Gamma(2s+1)}{\Gamma(s+1)}[1 / (e^{2b}-1)+4^s]. 
\end{align}
Here, the inequality is due to the fact $(2k+ 1/2)^{-2s} \le 1$ for $k \ge 1$. Note that for $b \ge \epsilon$, where $\epsilon > 0$ is arbitrary, the upper bound of \eqref{eq:inequality} becomes
\begin{align*}
(1+e^{-b})2^{-s} \frac{\Gamma(2s+1)}{\Gamma(s+1)}[1 / (e^{2b}-1)+4^s] \le (1+e^{-\epsilon})2^{-s} \frac{\Gamma(2s+1)}{\Gamma(s+1)}[1 / (e^{2\epsilon}-1)+4^s].
\end{align*}
Thus, combining \eqref{eq:posupperbound} with the above result, from \eqref{eq:expectomega3} we have for $b \ge \epsilon$,
\begin{align}
\label{eq:lemmaexpectomega}
\E(\omega^{-s}) \le 2^s b^s + f_{0} + L(s,\epsilon),
\end{align}
where $L(s,\epsilon) = (1+e^{-\epsilon})2^{-s} \frac{\Gamma(2s+1)}{\Gamma(s+1)}[1 / (e^{2\epsilon}-1)+4^s]$.

\indent Now, we consider $0<b < \epsilon$. Let $k(b)\equiv\E(\omega^{-s})$, where $\omega \sim PG(1,b)$. Then,
\begin{align}
\label{eq:limk}
\lim_{b \to 0} k(b) = \lim_{b \to 0}\cosh(b/2) \lim_{b \to 0} \int_{0}^{\infty} j(b,x) dx  = \lim_{b \to 0} \int_{0}^{\infty} j(b,x) dx,
\end{align}
where 
\begin{align*}
j(b,x) =  x^{-s-\frac{3}{2}} \sum_{\ell=0}^{\infty}(-1)^\ell \frac{(2\ell+1)}{\sqrt{2\pi }}\exp\Big[-\frac{(2\ell+1)^{2}}{8x}-\frac{b^{2}}{2}x \Big].
\end{align*}
Note that
\begin{align*}
j(b,x) \le (x^{-1}+1)x^{-\frac{3}{2}} \sum_{\ell=0}^{\infty}(-1)^\ell \frac{(2\ell+1)}{\sqrt{2\pi }}\exp\Big[-\frac{(2\ell+1)^{2}}{8x} \Big] = j(x), \,\text{say}.
\end{align*}
From \eqref{eq:fubini}, it follows that $\int_{0}^{\infty} j(x) dx \le 8C + 1$. Then by the DCT, from \eqref{eq:limk}, we have
\begin{align*}
\lim_{b \to 0} k(b) = \lim_{b \to 0} \int_{0}^{\infty} j(b,x) dx= k(0).
\end{align*}
So, $k(b) = \E (\omega^{-s})$ is continuous at $b=0$. Recall that $\E(\omega^{-s}) \le 8C + 1$ for $b=0$ and $0<s \le 1$. Thus, there exists some small $\epsilon>0$ such that $\E (\omega^{-s}) \le 8C + 2$ for $0<b<\epsilon$. Combining this result with \eqref{eq:lemmaexpectomega}, we have $\E(\omega^{-s}) \le 2^s b^s + L(s)$, where $L(s) = \max\{f_{0} + L(s,\epsilon),8C+2\}$. Thus, part $1$ is proved for $a = 1 $.

 Next, we prove the conclusion for $a>1$. From \cite{polson2013bayesian}, when $\omega \sim PG(a,b)$, we have
\begin{align}
\label{eq:omega_gamma}
\omega \overset{d} = \frac{1}{2\pi^2} \sum_{\ell=1}^{\infty} \frac{g_{\ell}}{(\ell-1/2)^2 +b^2/(4\pi^2)},
\end{align}
where $g_{\ell}$'s are mutually independent Gamma$(a, 1)$ random variables. Since $a > 1$, $g_{\ell} \overset{d} =  \tilde{g}_{\ell} + g^*_{\ell}$, where $\tilde{g}_{\ell}$ and $g^*_{\ell}$ are independent random variables following Gamma $(a-1,1)$ and Gamma $(1,1)$, respectively. Let $x_{1} = (1/[2\pi^2]) \sum_{\ell=1}^{\infty} g^*_{\ell}/[(\ell-1/2)^2 +b^2/(4\pi^2)]$. Then, $x_{1} \sim PG(1,b)$. Thus, we have $\E(x_{1}^{-s}) \le 2^s b^s + L(s)$. Since for $0 < s \le 1$, $\E (\omega^{-s}) \le \E (x_{1}^{-s})$, the same conclusion follows for $\omega \sim PG(a,b)$, where $a>1$. Thus, the proof for part $1$ is complete. 

Next, we prove part $2$. From \eqref{eq:omega_gamma}, we have
\begin{align*}
\E \,\omega^{-s} &= \E \,\Big[\frac{1}{2\pi^2} \sum_{\ell=1}^{\infty} \frac{g_{\ell}}{(\ell-1/2)^2+b^2/(4\pi^2)} \Big]^{-s}\\ &\le \E \,\Big[\frac{1}{2\pi^2}  \frac{g_{1}}{(1-1/2)^2+b^2/(4\pi^2)} \Big]^{-s} \\ &= \Big(\frac{\pi^2 + b^2}{2}\Big)^s \int_{0}^{\infty} g_{1}^{-s}\frac{1}{\Gamma(a)} g_{1}^{a-1} \exp{(-g_{1})}\, d \,g_{1} \\ &= 2^{-s} (\pi^2 + b^2)^s \frac{\Gamma(a-s)}{\Gamma(a)}.
\end{align*}
Thus, the proof for part $2$ is complete.
\end{proof}

\begin{lemma}
\label{le:e2omega}
If $\omega \sim PG(a,b), \, a >0, \, b \ge 0$, then $\E\,\omega \le a/4$.
\end{lemma}
\begin{proof}
From \eqref{eq:omega_gamma}, we have
\begin{align}
\label{eq:ineqalityomega}
\E\, \omega  &= \E \Big[\frac{1}{2\pi^2} \sum_{\ell=1}^{\infty} \frac{g_{\ell}}{(\ell-1/2)^2+b^2/(4\pi^2)} \Big] \nonumber\\ 
&\le \E \Big[\frac{1}{2\pi^2} \sum_{\ell=1}^{\infty} \frac{g_{\ell}}{(\ell-1/2)^2} \Big]\nonumber\\ 
&= \frac{1}{2\pi^2} \sum_{\ell=1}^{\infty} \frac{a}{(\ell-1/2)^2} = a \E \, \omega_{2},
\end{align}
where $\omega_{2} \sim PG(1,0)$. By \cite{polson2013bayesian}, $PG(1,0) = J^*(1,0)/4$. From \cite{devroye2009exact}, the density for $J^*(1,0)$ is 
\begin{align*}
f^*(x) = \pi \sum_{\ell=0}^{\infty} (-1)^\ell(\ell+ 1/2) \exp\Big[ -\frac{(\ell+1/2)^2 \pi^2 x}{2} \Big],
\end{align*}
then,
\begin{align*}
\E \,J^*(1,0) &= \int_{0}^{\infty} x \, \pi \sum_{\ell=0}^{\infty} (-1)^\ell(\ell+ 1/2) \exp\Big[ -\frac{(\ell+1/2)^2 \pi^2 x}{2} \Big] dx \\
&= \sum_{\ell=0}^{\infty} \int_{0}^{\infty} x \, \pi  (-1)^\ell(\ell+ 1/2) \exp\Big{[} -\frac{(\ell+1/2)^2 \pi^2 x}{2} \Big{]} dx\\
&= \sum_{\ell=0}^{\infty} \pi  (-1)^\ell(\ell+ 1/2) \Gamma(2) \Big(\frac{1}{2} \Big(\ell+ \frac{1}{2} \Big)^2 \pi^2 \Big)^{-2} \\&= \frac{4}{\pi^3} \sum_{\ell=0}^{\infty} \frac{(-1)^{\ell}}{(\ell+1/2)^3}= \frac{4}{\pi^3} \Phi\Big(-1,3,\frac{1}{2}\Big)= \frac{4}{\pi^3} \frac{\pi^3}{4}=1,
\end{align*}
where $\Phi(\cdot)$ is the Lerch transcendent function and $\Phi(-1,3,\frac{1}{2}) = \pi^3/4$. Note that, we achieve the second equality above by applying the following logic. Let $h_{3}(x,\ell) = x \, \pi  (-1)^\ell(\ell+ \frac{1}{2}) \exp\big{[} -\frac{(\ell+1/2)^2 \pi^2 x}{2} \big{]}$, then 
\begin{align*}
\sum_{\ell=0}^{\infty} \int_{0}^{\infty} \abs{h_{3}(x,\ell)} dx & = \sum_{\ell=0}^{\infty} \int_{0}^{\infty} x \, \pi  \Big(\ell+ \frac{1}{2}\Big) \exp\Big{[} -\frac{(\ell+1/2)^2 \pi^2 x}{2} \Big{]} dx\\ & = \frac{4}{\pi^3} \sum_{\ell=0}^{\infty} \frac{1}{(\ell+1/2)^3} < \infty, 
\end{align*}
Hence, $h_{3}(x,\ell)$ is integrable with respect to the product measure of the counting measure and the Lebesgue measure. Thus, by Fubini's Theorem, the second equality follows. Consequently, $\E\omega_{2} = \E\, J^*(1,0)/4= 1/4$. From \eqref{eq:ineqalityomega}, it follows $\E\, \omega \le a/4$.
\end{proof}

\begin{remark} \cite{wang2018geometric} proved Lemma \ref{le:eomega} in the special case when $a=1$. Although their result is correct as stated, their proof has an error which can be repaired following the techniques used in the proof of Lemma \ref{le:eomega} here. Lemma \ref{le:e2omega} for $a=1$ is also proved in \cite{wang2018geometric}. 
\end{remark}

\section{Proof of Theorem \ref{theoremimproper}}
\begin{proof} We first prove the geometric ergodicity of the $\tilde{\Psi}$ chain by establishing a drift condition. We consider the following drift function
\begin{align}
\label{eq:theorem2driftfunction}
V(\eta) = \sum_{i=1}^{n}\abs{ m_{i}^\top \eta}   + \sum_{j=1}^{r}(u_{j}^\top u_{j})^{-c},
\end{align}
where $c \in (0,1/2)$ to be determined later.

Since $M$ has full rank, $V(\eta): \mathbb{R}^{p+q}\backslash N \rightarrow[0,\infty)$ is unbounded off compact sets. We prove that, for any $\eta,\eta' \in \mathbb{R}^{p+q}\backslash N$, there exist constants $\rho \in [0,1)$ and $L>0$ such that
\begin{align}
\label{eq:theorem2driftcondition}
\E[V(\eta) \mid \eta\prime] = \E\{\E[V(\eta) \mid  \omega, \tau, y]\mid \eta\prime,y\} \le \rho V(\eta\prime) + L.
\end{align}
The first term in the drift function is
$\sum_{i=1}^{n}\abs{m_{i}^\top \eta} = l^\top M \eta$, where
$l = (l_{1},l_{2},...,l_{n})$ is defined as $l_{i} = 1$ if
$m_{i}^\top \eta \ge 0, l_{i} = -1$ if $m_{i}^\top \eta < 0$. Since
$Q=0$, from \eqref{eq:distrieta} we have
\begin{align}
\label{eq:theorem2term1a}
&\E \biggl[\sum_{i=1}^{n} \abs{ m_{i}^\top \eta} \bigg{|}  \omega, \tau, y \biggr] \nonumber\\ &=  l^\top M (M^\top \Omega(\omega) M + B(\tau))^{-1} M^\top \kappa \nonumber\\
&\le  \sqrt{l^\top M (M^\top \Omega(\omega) M + B(\tau))^{-1} M^\top l}  \sqrt{\kappa^\top M (M^\top \Omega(\omega) M  + B(\tau))^{-1} M^\top \kappa} \nonumber\\
& \le  \sqrt{l^\top M (M^\top \Omega(\omega) M )^{-1} M^\top l}  \sqrt{\kappa^\top M (M^\top \Omega(\omega) M  )^{-1} M^\top \kappa}  \nonumber\\
&=  \sqrt{l^\top \Omega(\omega)^{-1/2} P_{\Omega(\omega)^{1/2}M} \Omega(\omega)^{-1/2} l}  \sqrt{\kappa^\top M (M^\top \Omega(\omega) M  )^{-1} M^\top \kappa}  \nonumber\\
&\le  \sqrt{\sum_{i=1}^{n} \frac{1}{\omega_{i}}} \sqrt{\kappa^\top M (M^\top \Omega(\omega) M  )^{-1} M^\top \kappa},
\end{align}
where the first inequality follows from the Cauchy-Schwarz inequality, $P_{\Omega(\omega)^{1/2}M} \equiv \Omega(\omega)^{1/2} \linebreak M (M^\top \Omega(\omega) M)^{-1} M^\top \Omega(\omega)^{1/2}$ is a projection matrix, and the third inequality follows from the fact that $\I \succeq P_{\Omega(\omega)^{1/2}M}$. Recall that $k_{i} = y_{i} - 1/2, \, i = 1,...,n$. Define, $v_{i}=-2k_{i} m_{i}$ as the $i^{th}$ row of an $n \times (p+q)$ matrix $V$. Note that $v_{i}v_{i}^\top =m_{i}m_{i}^\top, \, i = 1,...,n$. Since the conditions $3$ and $4$ in Theorem \ref{theoremimproper} are in force, by Lemma $3$ in \cite{wang2018geometric}, for the second part of \eqref{eq:theorem2term1a}, we have
\begin{align}
\label{eq:theorem2term1d}
\sqrt{\kappa^\top M(M^\top \Omega(\omega) M)^{-1}M^\top\kappa}=\sqrt{\frac{1}{4} 1^\top V(V^\top \Omega(\omega) V)^{-1}V^\top 1} \le \sqrt{\frac{\rho_{1}}{4} \sum_{i=1}^{n} \frac{1}{\omega_{i}}},
\end{align}
where $\rho_{1} \in [0,1)$ is a constant. Applying \eqref{eq:theorem2term1d} to \eqref{eq:theorem2term1a}, we have
\begin{align}
\label{eq:theorem2term1}
\E \bigg[\sum_{i=1}^{n}\abs{ m_{i}^\top\eta} \bigg{|}  \omega, \tau, y \bigg] \le \frac{\sqrt{\rho_{1}}}{2}   \sum_{i=1}^{n} \frac{1}{\omega_{i}}.
\end{align}
\indent Next, we consider the inner expectation of the second term in the drift function \eqref{eq:theorem2driftfunction}. Note that, for $c \in (0,1/2)$, we have
\begin{align}
\label{eq:theorem2term2a}
  &\E\Big[(u_{j}^\top u_{j})^{-c} \Big{|} \omega, \tau, y\Big]\nonumber\\  &=  \bigg(\sum_{i=1}^{n}\omega_{i} \tr(Z^\top Z) + \tau_{j}\bigg)^{c} \E\bigg[\bigg\{u_{j}^\top  \bigg(\sum_{i=1}^{n}\omega_{i} \tr(Z^\top Z) + \tau_{j}\bigg)\I_{q_{j}} u_{j}\bigg\}^{-c} \bigg{|} \omega, \tau, y\bigg] \nonumber \\
  &\le \bigg[\bigg(\sum_{i=1}^{n}\omega_{i} \tr(Z^\top Z)\bigg)^{c} + \tau_{j}^{c}\bigg]\E\bigg[\bigg\{u_{j}^\top  (R_{j} \tilde{S}^{-1} {R_{j}}^\top)^{-1} u_{j}\bigg\}^{-c} \bigg{|} \omega, \tau, y\bigg]  \nonumber\\
  &\le \frac{2^{-c} \Gamma(-c + q_{j}/2)}{\Gamma(q_{j}/2)} \bigg[\bigg(\sum_{i=1}^{n}\omega_{i} \tr(Z^\top Z)\bigg)^{c} + \tau_{j}^{c} \bigg],
\end{align}
where the first inequality follows from Lemma \ref{lemma:rineq} and the fact that $(a+b)^s \le a^s +b^s$ for $a>0, \, b>0,$ and $0 \le s <1$. For the last inequality, note that, by \eqref{eq:udistribution}, we have $u_{j} \mid \omega,\tau,y \sim N(R_{j}(-\tilde{S}^{-1}\tilde{R}^\top X^\top\kappa+\tilde{S}^{-1}Z^\top\kappa), R_{j}\tilde{S}^{-1} {R_{j}}^\top)$ where $R_{j}$ is a $q_{j} \times q$ matrix consisting of $0$'s and $1$'s such that $R_{j}u = u_{j}$. Thus, given $\omega,\tau,y$, $(R_{j}\tilde{S}^{-1}{R_{j}}^\top)^{-\frac{1}{2}}u_{j}$ has a multivariate normal distribution with the identity covariance matrix. Hence, conditional on $\omega, \tau, y$, the distribution of $u_{j}^\top (R_{j}\tilde{S}^{-1} {R_{j}}^\top)^{-1}u_{j}$ is $\chi_{q_{j}}^{2}(w)$, where $w$ is some (unimportant) noncentrality parameter and $q_{j}$ is the degrees of freedom for the relevant Chi-square distribution. Therefore, by Lemma 4 in \cite{roman2012convergence}, we have
\begin{align*}
\E\Big[\Big\{u_{j}^\top  (R_{j}\tilde{S}^{-1}{R_{j}}^\top)^{-1} u_{j}\Big\}^{-c} \Big{|} \omega, \tau, y\Big]
\le \frac{2^{-c} \Gamma(-c + q_{j}/2)}{\Gamma(q_{j}/2)}.
\end{align*}
Applying the above result, the inequality in \eqref{eq:theorem2term2a} is obtained.

Combining \eqref{eq:theorem2term1} and \eqref{eq:theorem2term2a}, from \eqref{eq:theorem2driftfunction}, we have
\begin{align}
\label{eq:theorem2wholeinner}
E[V(\eta) \mid \omega, \tau, y]  \le  \frac{\sqrt{\rho_{1}}}{2}   \sum_{i=1}^{n} \frac{1}{\omega_{i}}+\sum_{j=1}^{r} \frac{2^{-c} \Gamma(-c + q_{j}/2)}{\Gamma(q_{j}/2)}\Big[ \Big(\sum_{i=1}^{n}\omega_{i} \tr(Z^\top Z)\Big)^{c} + \tau_{j}^{c}\Big].
\end{align}
Next, we consider the outer expectation in \eqref{eq:theorem2driftcondition}. By Lemma \ref{le:eomega}, we have
\begin{align}
\label{eq:theorem2outerterm1}
 \E \Big[\frac{\sqrt{\rho_{1}}}{2}  \sum_{i=1}^{n} \frac{1}{\omega_{i}}\mid \eta^{^{\prime}}, y\Big] &\le \Big[2\sum_{i=1}^{n} \abs{m_{i}^\top \eta'} + nL(1)  \Big] \frac{ \sqrt{\rho_{1}}}{2}  \nonumber\\&= \sqrt{\rho_{1}}  \sum_{i=1}^{n} \abs{m_{i}^\top \eta'} +  \frac{ nL(1)\sqrt{\rho_{1}}}{2}.
\end{align}
For the outer expectation of the other terms on the right hand side of \eqref{eq:theorem2wholeinner}, we now consider the expectation for $\tau_{j}^{c}$. Recall from section \ref{sec:blockgibbs} that $\tau_{j} \mid \eta^{\prime},y \overset{ind}\sim $ Gamma$(a_{j} + q_{j}/2,b_{j} + u_{j}'^\top  u_{j}'/2)$, $j = 1,...,r$. Then, it follows that 
\begin{align*}
\E[\tau_{j}^{c} \mid \eta^{\prime},y] =\frac{\Gamma(a_{j} + \frac{q_{j}}{2} + c)}{\Gamma(a_{j}+\frac{q_{j}}{2})}\Big(b_{j}+\frac{1}{2} u_{j}'^\top u_{j}'\Big)^{-c}.
\end{align*}
Define, $G_{j}(-c) = 2^{c} \Gamma(a_{j} +q_{j}/2 + c)/\Gamma(a_{j} + q_{j}/2)$, $j = 1,2,...,r$. Hence,
\begin{align}
\label{eq:expectau}
\E[\tau_{j}^{c} \mid \eta^{\prime},y] & =2^{-c}G_{j}(-c)\Big(b_{j}+ \frac{u_{j}'^\top u_{j}'}{2}\Big)^{-c} \nonumber \\  & \le G_{j}(-c)[(2b_{j})^{-c}\I_{(0,\infty)}(b_{j}) + (u_{j}'^\top u_{j}')^{-c} \I_{\{0\}}(b_{j})].
\end{align}
Also,
\begin{align}
\label{eq:tauinequality}
\sum_{j=1}^{r} \frac{2^{-c} \Gamma(q_{j}/2-c)}{\Gamma(q_{j}/2)}G_{j}(-c)(u_{j}'^\top u_{j}')^{-c} \I_{\{0\}}(b_{j}) \le \delta_{1}(c)\sum_{j=1}^{r} (u_{j}'^\top u_{j}')^{-c}
\end{align}
where $\delta_{1}(c) = 2^{-c} \max_{j \in A} \frac{\Gamma(q_{j}/2-c)}{\Gamma(q_{j}/2)}G_{j}(-c) \ge 0$. Recall that $A = \{j \in \{1,2,...,r\}: b_{j} = 0\}$. From the condition $1$ of Theorem \ref{theoremimproper}, we have $a_{j} <0$ when $b_{j} = 0$ (i.e. when $j \in A$). From the proof of Proposition $2$ in \cite{roman2012convergence}, it follows that, there exists $c\in C \equiv (0,1/2) \cap (0,\tilde{a})$, where $\tilde{a} = - \max_{j \in A} a_{j}$, such that $\delta_{1}(c) < 1$.

Using \eqref{eq:theorem2outerterm1}, \eqref{eq:expectau}, and Jensen's inequality, from \eqref{eq:theorem2wholeinner}, we obtain
\begin{align}
  \E[V(\eta) \mid \eta\prime] &=\E\{\E[V(\eta \mid  \omega, \tau, y)]\mid \eta\prime,y\}\\ &\le \sqrt{\rho_{1}} \sum_{i=1}^{n} \abs{m_{i}^\top\eta'}  + \frac{nL(1) \sqrt{\rho_{1}}}{2}+ \sum_{j=1}^{r} \frac{2^{-c} \Gamma(-c + q_{j}/2)}{\Gamma(q_{j}/2)} \nonumber\\& \quad
   \times\big\{\big{(}\tr(Z^\top Z) \sum_{i=1}^{n} \E(\omega_{i} \mid \eta',y)\big{)}^{c} +  G_{j}(-c) \nonumber\\& \quad \quad \times \big[(2b_{j})^{-c}\I_{(0,\infty)}(b_{j}) + (u_{j}'^\top u_{j}')^{-c} \I_{\{0\}}(b_{j})\big] \big\} \nonumber \\
&\le \sqrt{\rho_{1}}   \sum_{i=1}^{n} \abs{m_{i}^\top \eta'}+ \delta_{1}(c)\sum_{j=1}^{r} (u_{j}'^\top u_{j}')^{-c}  +L \le \rho V(\eta\prime) + L,  \nonumber
\end{align}
where the second inequality is due to \eqref{eq:tauinequality} and Lemma \ref{le:e2omega}. Here, $\rho=\max\{\sqrt{\rho_{1}}, \delta_{1}(c)\}$, and 
\begin{align*}
L= \frac{n L(1) \sqrt{\rho_{1}}}{2} &+ \sum_{j=1}^{r} \frac{2^{-c} \Gamma(-c + q_{j}/2)}{\Gamma(q_{j}/2)} \Big(\frac{\tr(Z^\top Z)n}{4}\Big)^{c}\\ &+ \sum_{j=1}^{r} \frac{2^{-c} \Gamma(-c + q_{j}/2)}{\Gamma(q_{j}/2)} G_{j}(-c)(2b_{j})^{-c}\I_{(0,\infty)}(b_{j}).
\end{align*} 
Recall that $\rho_{1}, \delta_{1}(c) \in [0,1)$, thus $\rho <1$. Consequently, \eqref{eq:theorem2driftcondition} holds. In addition, we can show that $\tilde{\Psi}$ chain is a Feller Markov chain by the following steps. % (using the technique in Appendix C of \cite{wang2018convergence})
Let $K(\eta', \cdot)$ denote the Mtf corresponding to \eqref{eq:mtd_eta}. To prove that the $\tilde{\Psi}$ chain is a Feller Markov chain is to show that $K(\eta', A)$ is a lower semi-continuous function on $\mathbb{R}^{p+q}\backslash N$ for each fixed open set $A$ in $\mathbb{R}^{p+q}\backslash N$. For a sequence $\{\eta'_{m}\}$, using \eqref{eq:mtd_eta}, Fatou's Lemma and independence of the conditional distribution of $\omega$ and $\tau$ given $(\eta',y)$, we have
\begin{align*}
\liminf_{m\rightarrow \infty}K(\eta'_{m}, A) &= \liminf_{m\rightarrow \infty} \int_{A}  \tilde{k}(\eta \mid \eta'_{m}) d\eta\\
&= \liminf_{m\rightarrow \infty} \int_{A}  \int_{\mathbb{R}_{+}^{r}} \int_{\mathbb{R}_{+}^{n}} \pi(\eta \mid \omega,\tau,y) \pi(\omega,\tau \mid \eta'_{m},y)d \omega \,d\, \tau \,d\eta\\
&\ge \int_{A}  \int_{\mathbb{R}_{+}^{r}} \int_{\mathbb{R}_{+}^{n}}  \pi(\eta \mid \omega,\tau,y) \liminf_{m\rightarrow \infty} \pi(\omega,\tau \mid \eta'_{m},y) d \omega \,d\, \tau \,d\eta\\
&= \int_{A}  \int_{\mathbb{R}_{+}^{r}} \int_{\mathbb{R}_{+}^{n}}  \pi(\eta \mid \omega,\tau,y) \liminf_{m\rightarrow \infty} [\pi(\omega \mid \eta'_{m},y) \pi(\tau \mid \eta'_{m},y)] d \omega \,d\, \tau \,d\eta.
\end{align*}
Considering the conditions $1$ and $2$ in Theorem \ref{theoremimproper}, for any fixed $(\eta, \omega, y)$, both $\pi(\omega \mid \eta'_{m},y)$ and $\pi(\tau \mid \eta'_{m},y)$ are continuous functions on $\mathbb{R}^{p+q}\backslash N$. Hence, if $\eta'_{m} \rightarrow \eta'$, then $\liminf_{m\rightarrow \infty}K(\eta'_{m}, A) \ge K(\eta', A)$, and we can conclude that the $\tilde{\Psi}$ chain is a Feller Markov chain. Thus, by Lemma $15.2.8$ in \cite{meyn1993markov}, GE of the $\tilde{\Psi}$ chain is proved.

Next, using the similar techniques as in \cite{wang2018convergence}
and \cite[][Lemma 12]{roman2012thesis}, we can establish that the GE of the original chain $\Psi$ follows from that of $\tilde{\Psi}$. We include a proof here for completeness. Let
$\mathsf{X} \equiv \mathbb{R}^{p+q}, \, \tilde{\mathsf{X}} \equiv
\mathbb{R}^{p+q}\backslash N$. Let $K$ and $\tilde{K}$ denote the Mtfs
of $\Psi$ and $\tilde{\Psi}$ chains respectively. Also, since the
Lebesgue measure of $N$ is zero, $\tilde{K}(x,B) = K(x,B)$, for any
$x \in \tilde{\mathsf{X}}$ and
$B \in \B_{\tilde{\mathsf{X}}} = \{ \tilde{\mathsf{X}}\cap A: A\in \B_{\mathsf{X}} \}$, where
$\B_{\mathsf{X}}$ denotes the Borel $\sigma$-algebra of $\mathbb{R}^{p+q}$ and
$\B_{\tilde{\mathsf{X}}}$ denotes the Borel $\sigma$-algebra of
$\mathbb{R}^{p+q}\backslash N$, respectively.

Let $\mu$ and $\tilde{\mu}$ be the Lebesgue measures on $\mathsf{X}$
and $\tilde{\mathsf{X}}$, respectively. As the Mtds are strictly
positive for the two chains, the $\Psi$ chain is $\mu$-irreducible and
the $\tilde{\Psi}$ chain is $\tilde{\mu}$-irreducible. Both chains are
aperiodic. Note that, $\mu$ and $\tilde{\mu}$ are also the maximal
irreducibility measures of $\Psi$ and $\tilde{\Psi}$ chains,
respectively. By Theorem $15.0.1$ in \cite{meyn1993markov}, since we
have established the GE of the $\tilde{\Psi}$ chain, there exists a
$v$-petite set $C \in \B_{\tilde{\mathsf{X}}}$,
$\rho_{C} < 1, \, M_{C} < \infty$ and $\tilde{K}^{\infty}(C) > 0$ such
that $\tilde{\mu}(C) > 0$ and
\begin{align}
\label{eq:absineq}
\abs{\tilde{K}^{m}(x,C) - \tilde{K}^{\infty}(C)} < M_{C} \rho_{C}^{m}, 
\end{align}
for all $x \in C$. Also, it can be shown that 
\begin{align}
\label{eq:mtfequa}
K^{m}(x, B) = \tilde{K}^{m}(x, B \cap \tilde{\mathsf{X}}),
\end{align}
for any $x \in \tilde{\mathsf{X}}$ and $B \in \B_{\mathsf{X}}$. Note
that, $K^{m}$ and $\tilde{K}^{m}$ indicate the corresponding $m$-step
Mtfs. Thus, for $x \in C$, $K ^{m}(x, C) = \tilde{K}^{m}(x, C)$. Then,
\eqref{eq:absineq} becomes
%\begin{align*}
$\abs{K^{m}(x,C) - \tilde{K}^{\infty}(C)} < M_{C} \rho_{C}^{m}.$
%\end{align*}
Since $\mu(N) = 0$, we have $\mu(C) = \tilde{\mu}(C)$. Recall that
$\tilde{\mu}(C) >0$, thus $\mu(C)>0$. Note that, $C$ is $v$-petite for
the $\tilde{\Psi}$ chain, then for all $x \in C$ and
$B \in \B_{\tilde{\mathsf{X}}}$,
\begin{align}
\label{eq:petitedefi}
\sum_{m=0}^{\infty} \tilde{K}^{m}(x, B) a(m) \ge v(B),
\end{align}
where $v$ is a nontrivial measure on $\B_{\tilde{\mathsf{X}}}$ and $a(m)$ is a mass function on $\{0,1,2,...\}$. It can be shown that a nontrivial measure on $\B_{\mathsf{X}}$, which is 
\begin{align}
\label{eq:vstar}
v^*(\cdot) = v(\cdot \cap \tilde{\mathsf{X}}),
\end{align}
is well defined. Then for any $x \in C$ and any $B \in \B_{\mathsf{X}}$, using \eqref{eq:mtfequa}, \eqref{eq:petitedefi} and \eqref{eq:vstar}, we have
\[
  \sum_{m=0}^{\infty} K^{m}(x, B) a(m) =\sum_{m=0}^{\infty}
  \tilde{K}^{m}(x, B \cap \tilde{\mathsf{X}}) a(m) \ge v(B \cap \tilde{\mathsf{X}})=
  v^*(B).\] Hence, $C$ is also a petite set for the $\Psi$
chain. Applying Theorem $15.0.1$ in \cite{meyn1993markov} again, GE
of the $\Psi$ chain is proved. Hence, we show that GE of
$\tilde{\Psi}$ implies that of the original chain $\Psi$.

\end{proof}

\end{appendix}

\bibliographystyle{ims}
\bibliography{Ref}

%\begin{acknowledgements}
%\end{acknowledgements}



\end{document}